\documentclass[smallextended,referee,envcountsect,]{svjour3}

\smartqed

\usepackage[dvipsnames]{xcolor}  

\usepackage{hyperref}
\usepackage{url}

\usepackage[utf8]{inputenc} % allow utf-8 input
\usepackage[T1]{fontenc}    % use 8-bit T1 fonts
\usepackage{hyperref}       % hyperlinks
\usepackage{url}            % simple URL typesetting
\usepackage{booktabs}       % professional-quality tables
\usepackage{amsfonts}       % blackboard math symbols
\usepackage{nicefrac}       % compact symbols for 1/2, etc.
\usepackage{microtype}      % microtypography
\usepackage{graphicx}
\usepackage{wrapfig}
\usepackage{makecell}
\usepackage{multirow}
\usepackage{subcaption}
\usepackage[nointegrals]{wasysym}

\usepackage[ruled,vlined]{algorithm2e}

\usepackage{amsmath,amssymb}

\usepackage{enumitem}

\journalname{JOTA}

\everypar{\looseness=-1}
\allowdisplaybreaks

\begin{document}

\title{An Optimal Transport Perspective on \\Unpaired Image Super-Resolution}

\author{Milena Gazdieva$^{\star\dagger}$\thanks{$^\star$ Skolkovo Institute of Science and Technology,  
Moscow, Russia \\ $^\times$ Indian Space Research Organization, India\\ $^\dagger$ Artificial Intelligence Research Institute, Moscow, Russia \\ $^\ddagger$ University of Oxford, Oxford, UK \\ $^\circ$ AI Foundation and Algorithm Lab, Moscow, Russia} \and Petr Mokrov$^\star$ \and Litu Rout$^\times$ \and Alexander Korotin$^{\star \, \dagger}$ \and Andrey Kravchenko$^\ddagger$ \and Alexander Filippov$^\circ$ \and Evgeny Burnaev$^{\star\, \dagger}$
}
\authorrunning{Milena Gazdieva et al.}

\institute{Milena Gazdieva, Corresponding author \at
             Skolkovo Institute of Science and Technology \\
             Artificial Intelligence Research Institute\\
              Moscow, Russia\\
              milena.gazdieva@skoltech.ru
}

\date{Received: date / Accepted: date}

\maketitle

\begin{abstract}
\looseness=-1
\vspace{-1mm}
Real-world image super-resolution (SR) tasks often do not have paired datasets, which limits the application of supervised techniques. As a result, the tasks are usually approached by \textit{unpaired} techniques based on Generative Adversarial Networks (GANs), which yield complex training losses with several regularization terms, e.g., content or identity losses. While GANs usually provide good practical performance, they are used heuristically, i.e., theoretical understanding of their behaviour is yet rather limited. We theoretically investigate optimization problems which arise in such models and find two surprising observations. First, the learned SR map is always an \textit{optimal transport} (OT) map. Second, we theoretically prove and empirically show that the learned map is \textit{biased}, i.e., it does not actually transform the distribution of low-resolution images to high-resolution ones. Inspired by these findings, we investigate recent advances in neural OT field to resolve the \textit{bias} issue. We establish an intriguing connection between regularized GANs and neural OT approaches. We show that unlike the existing GAN-based alternatives, these algorithms aim to learn an \textit{unbiased} OT map. We empirically demonstrate our findings via a series of synthetic and real-world unpaired SR experiments. Our source code is publicly available at \url{https://github.com/milenagazdieva/OT-Super-Resolution}.
\end{abstract}

\vspace{2mm}
\noindent Communicated by Martin Takac.

\keywords{Optimal transport, generative modeling}
 
% \newpage
%\vspace{-3mm}
\section{Introduction}
\vspace{-2mm}
\begin{wrapfigure}{r}{0.45\textwidth}
  \vspace{-5mm}\begin{center}
    \includegraphics[width=\linewidth]{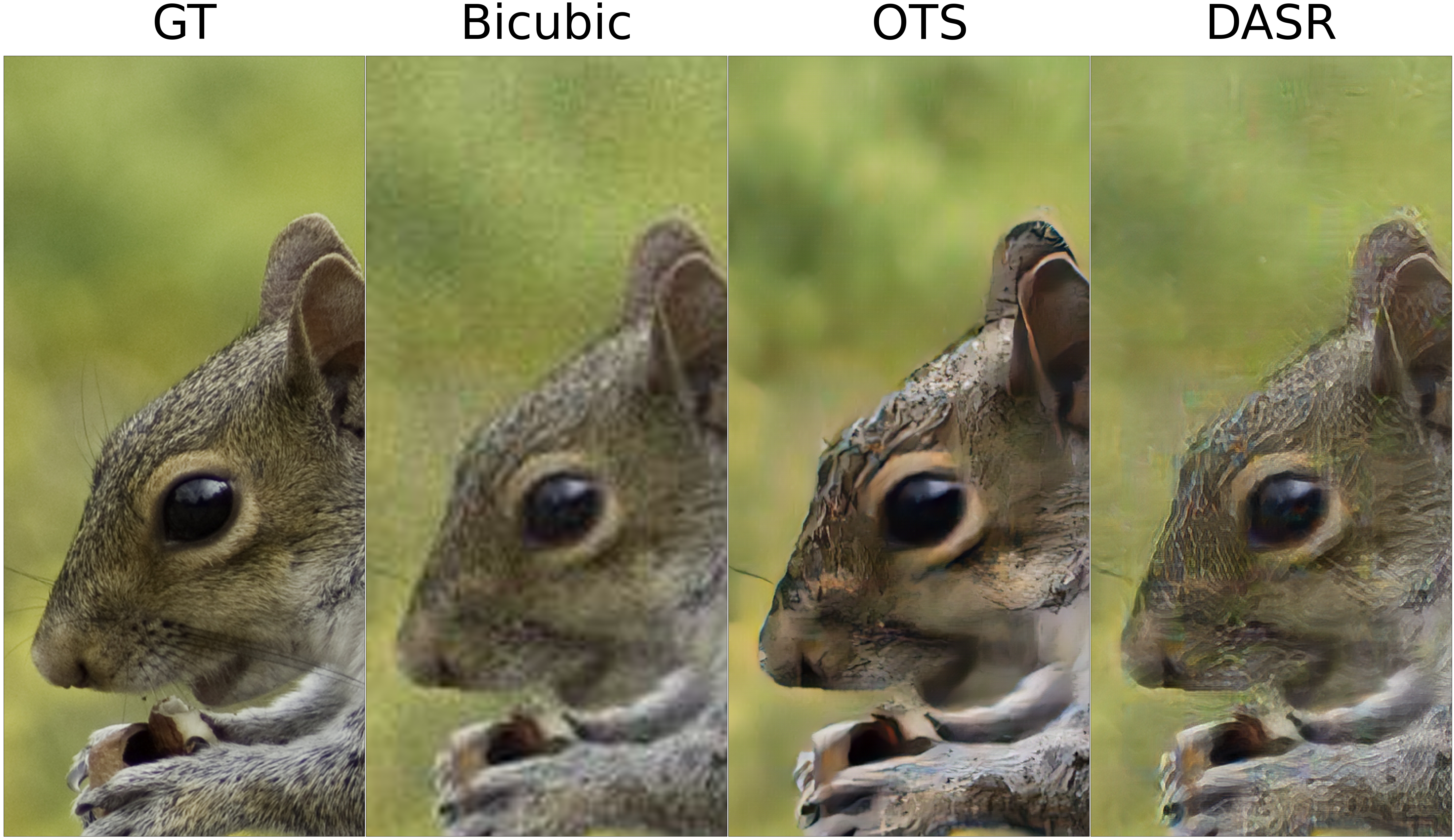}
  \end{center}
  \vspace{-5mm}
  \caption{\centering Super-resolution of a squirrel using Bicubic upsample, OTS and DASR \cite{Wei_2021_CVPR} methods (4$\times$4 upsample, 370$\times$800 crops).}
  \vspace{-5mm}
\end{wrapfigure}
The problem of image super-resolution (SR) is to reconstruct a high-resolution (HR) image from its low-resolution (LR) counterpart. In many modern deep learning approaches, SR networks are trained in a supervised manner by using synthetic datasets containing LR-HR \textit{pairs} \cite[\wasyparagraph 4.1]{lim2017enhanced};
\cite[\wasyparagraph 4.1]{zhang2018rcan}. For example, it is common to create LR images from HR with a simple downscaling, e.g., bicubic \cite[\wasyparagraph 3.2]{ledig2017photo}.
However, such an artificial setup barely represents the practical setting, in which  the degradation is more sophisticated and unknown \cite{maeda2020unpaired}. This obstacle necessitates developing methods capable of learning SR maps from \textit{unpaired} data without considering prescribed degradations. Currently, the problem of unpaired image SR is typically solved by adversarial methods \cite{fritsche2019frequency,Wei_2021_CVPR,zhou2020guided,liu2023unpaired}, making them the primary focus of our research. Recently, diffusion-based approaches to image SR have emerged \cite{saharia2022image,yue2023resshift,pmlr-v202-liu23ai}, but they are targeted at the \textit{paired} image SR setup. To the best of our knowledge, all existing diffusion-based SR methods are \textit{paired}, i.e., they require paired samples when learning the diffusion.

\noindent\textbf{Contributions.} We study the unpaired image SR task and its solutions based on Generative Adversarial Networks \cite[GANs]{goodfellow2014generative} and analyse them from the Optimal Transport \cite[OT]{villani2008optimal} perspective.

\begin{enumerate}[leftmargin=*]
\looseness=-1
\item \textbf{Theory I.} We investigate the GAN optimization objectives regularized with content losses, which are common in unpaired image SR methods (\wasyparagraph\ref{sec-biased-ot}). We prove that the solution to such objectives is always an optimal transport map which is, in general, biased.

\item \textbf{Theory II.} We explain the ideas that stand behind recent algorithms from the field of neural OT \cite{korotin2023neural,fan2023neural} which aim to recover the true (unbiased) OT map. To do this, we show that their algorithms' optimization objective can be viewed as a certain particular case of GAN-based objectives regularized with content losses (\wasyparagraph\ref{sec-unbiased-ots}). We also establish connections between these algorithms and regularized GANs that use integral probability metrics (IPMs) as a loss (\wasyparagraph\ref{sec-gans-vs-ot}).

\item\textbf{Practice.}   We empirically show that oppositely to neural OT methods GANs' maps are \textit{biased} (\wasyparagraph\ref{sec-bias-experiments}), i.e., they do not transform the LR image distribution to the true HR image distribution. We demonstrate the findings on the synthetic (\wasyparagraph \ref{sec-bias-experiments}) and real-world (\wasyparagraph\ref{sec-aim-bias}) super-resolution task.
\end{enumerate}
\vspace{-1.5mm}
We emphasize the importance of the revealed bias issue of GANs which is critical, e.g., in medical applications. In MRI super-resolution, the biased models may generate inexistent details, potentially leading to misdiagnosis \cite{bissoto2021gan}.

\noindent\textbf{Notation.} We use $\mathcal{X}=\mathbb{R}^{D_x},\mathcal{Y}=\mathbb{R}^{D_y}$ to denote data spaces and $\mathcal{P}(\mathcal{X}), \mathcal{P}(\mathcal{Y})$ to denote the respective sets of probability distributions on them. We denote by $\Pi(\mathbb{P},\mathbb{Q})$ the set of probability distributions on $\mathcal{X} \times \mathcal{Y}$ with marginals $\mathbb{P}$ and $\mathbb{Q}$.
For a measurable map $T:\mathcal{X}\rightarrow\mathcal{Y}$, we denote the associated push-forward operator by $T_{\#}$. The expression $\|\cdot\|$ denotes the usual Euclidean norm if not stated otherwise. We denote the space of $\mathbb{Q}$-integrable functions on $\mathcal{Y}$ by $L^{1}(\mathbb{Q})$. Let $\mathbb{P}$ and $\mathbb{Q}$ be two distributions of LR and HR images, respectively, on spaces $\mathcal{X}$ and $\mathcal{Y}$, respectively.

\vspace{-2.5mm}
\section{Unpaired Image Super-Resolution Task}
\vspace{-2mm}

In this section, we formalize the \textit{unpaired} image super-resolution task that we consider (Figure \ref{fig:sr-setup}).

\begin{wrapfigure}{r}{0.55\textwidth}
  \vspace{-10mm}
  \begin{center}
    \includegraphics[width=\linewidth]{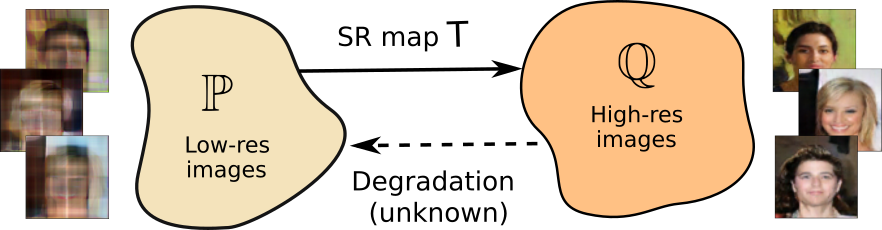}
  \end{center}
  \vspace{-3.5mm}
  \caption{\centering The task of super-resolution we consider. }
  \label{fig:sr-setup}
  \vspace{-9.0mm}
\end{wrapfigure}
  The learner has access to unpaired random samples from $\mathbb{P}$ and $\mathbb{Q}$. The task is to fit a map ${T:\mathcal{X}\rightarrow\mathcal{Y}}$ satisfying ${T_{\#}\mathbb{P}=\mathbb{Q}}$ which \textit{inverts} the degradation.

We highlight that the image SR task is theoretically ill-posed for two reasons.
\vspace{-1.5mm}
\begin{enumerate}[leftmargin=*]
    \item \textbf{Non-existence.} The degradation filter may be \textit{non-injective} and, consequently, \textit{non-invertible}. This is a theoretical obstacle to learn one-to-one SR maps $T$. 
    \item \textbf{Ambiguity.} There might exist \textit{multiple} maps satisfying ${T_{\#}\mathbb{P}=\mathbb{Q}}$ but only one inverting the degradation. With no prior knowledge about the correspondence between $\mathbb{P}$,  $\mathbb{Q}$, it is unclear how to pick this particular map.
\end{enumerate}
% \vspace{-1.2mm}

\textbf{The first issue} is usually not taken into account in practice. Most existing paired and unpaired SR methods learn one-to-one SR maps $T$, see \cite{ledig2017photo,lai2017deep,Wei_2021_CVPR}.

\textbf{The second issue} is typically softened by regularizing the model with the content loss.
In the real-world, it is reasonable to assume that HR and the corresponding LR images are close. Thus, the fitted SR map $T$ is expected to only \textit{slightly} change the input image. Formally, one may require the learned map $T$ to have the small value of
\vspace{-2mm}
\begin{equation}
    \mathcal{R}_{c}(T)\stackrel{def}{=}\int_{\mathcal{Y}}c\big(x,T(x)\big)d\mathbb{P}(x),
    \label{identity-c-loss}
\end{equation}

\vspace{-2.5mm}
where $c:\mathcal{X}\times\mathcal{Y}\rightarrow\mathbb{R}_{+}$ is a function  estimating   how different the inputs are. The most popular example is the $\ell^{1}$ \textit{identity} loss, i.e, formulation \eqref{identity-c-loss} for $\mathcal{X}=\mathcal{Y}=\mathbb{R}^{D}$ and $c(x,y)=\|x-y\|_{1}$.

More broadly, losses $\mathcal{R}_{c}(T)$
are typically called \textit{content} losses and incorporated into training objectives of methods for SR  \cite[\wasyparagraph 3.4]{lugmayr2019unsupervised}, \cite[\wasyparagraph 3]{kim2020unsupervised} and other unpaired tasks beside SR \cite[\wasyparagraph 4]{taigman2016unsupervised}, \cite[\wasyparagraph 5.2]{zhu2017unpaired} as regularizers. They stimulate the learned map $T$ to minimally change the image content.

A common approach to solve the unpaired SR via GANs is to define a loss function $\mathcal{D}:\mathcal{P}(\mathcal{Y})\times\mathcal{P}(\mathcal{Y})\rightarrow\mathbb{R}_{+}$ and train a generative neural network $T$ via minimizing
% \vspace{-1.5mm}
\begin{equation}\inf_{T:\mathcal{X}\mapsto\mathcal{Y}}\big[\mathcal{D}(T_{\#}\mathbb{P},\mathbb{Q})+\lambda \mathcal{R}_{c}(T)\big]. 
\label{base-gan-c}
\end{equation}

\looseness=-1
The term $\mathcal{D}(T_{\#}\mathbb{P},\mathbb{Q})$ ensures that the generated distribution $T_{\#}\mathbb{P}$ of SR images is close to the true HR distribution $\mathbb{Q}$. For convenience, we assume that $\mathcal{D}(\mathbb{Q}, \mathbb{Q})=0$ for all $\mathbb{Q}\in\mathcal{P(Y)}$. Two most popular examples of $\mathcal{D}$ are the Jensen--Shannon divergence \cite{goodfellow2014generative}, i.e., the vanilla GAN and the Wasserstein-1 loss \cite{arjovsky2017towards}. 

In unpaired SR methods, the optimization objectives are typically more complex than \eqref{base-gan-c}. In addition to the content or identity loss \eqref{identity-c-loss}, several other regularizations are usually introduced. Existing approaches to unpaired image SR mainly solve the problem in two steps. One group of approaches learn the degradation operation at the first step and then train a super-resolution model in a supervised manner using generated pseudo-pairs, see \cite{bulatyang2018learn,fritsche2019frequency}. Another group of approaches \cite{yuan2018unsupervised,maeda2020unpaired} firstly learn a mapping from real-world LR images to ``clean`` LR images, i.e., HR images, downscaled using predetermined (e.g., bicubic) operation, and then a mapping from ``clean`` LR to HR images. Most methods are based on CycleGAN \cite{zhu2017unpaired}, initially designed for the domain transfer task, and utilize cycle-consistency loss. Methods are also usually endowed with several other losses, e.g. content \cite[\wasyparagraph 3]{kim2020unsupervised}, identity \cite[\wasyparagraph 3.2]{wang2021unsupervised} or perceptual \cite[\wasyparagraph 3.4]{lugmayr2019unsupervised}. However, we emphasize that all methods have unpaired learning step which corresponds to the optimization objective \eqref{base-gan-c}. In Appendix \ref{gan-objectives-example}, we show that the \underline{learning objectives} of popular SR methods can be represented as \eqref{base-gan-c}.

\vspace{-2.5mm}
\section{Background on Optimal Transport}
\vspace{-2mm}

In this section, we give the key concepts of the OT theory \cite{villani2008optimal} that we use
in our paper.

\noindent\textbf{Primal form}. For two distributions ${\mathbb{P}\in\mathcal{P}(\mathcal{X})}$ and ${\mathbb{Q}\in\mathcal{P}(\mathcal{Y})}$ and a transport cost $c:\mathcal{X}\times\mathcal{Y}\rightarrow\mathbb{R}$, Monge's primal formulation of the \textit{optimal transport cost} is as follows: 
\begin{equation}
\text{Cost}(\mathbb{P},\mathbb{Q})\stackrel{\text{def}}{=}\inf_{T_{\#}\mathbb{P}=\mathbb{Q}}\ \int_{\mathcal{X}} c\big(x,T(x)\big)d\mathbb{P}(x),
\label{ot-primal-form-monge}
\end{equation}
where the minimum is taken over the  measurable functions (transport maps) $T:\mathcal{X}\rightarrow\mathcal{Y}$ 
that map $\mathbb{P}$ to $\mathbb{Q}$, see Figure \ref{fig:ot-map-def}. 
The optimal $T^{*}$ is called the \textit{optimal transport map}.

Note that \eqref{ot-primal-form-monge} is not symmetric, and this formulation does not allow mass splitting, i.e., for some
$\mathbb{P},\mathbb{Q}$ there may be no map $T$ that satisfies $T_{\#}\mathbb{P}=\mathbb{Q}$. Thus, \cite{kantorovitch1958translocation} proposed the relaxation:
\begin{equation}\text{Cost}(\mathbb{P},\mathbb{Q})\stackrel{\text{def}}{=}\inf_{\pi\in\Pi(\mathbb{P},\mathbb{Q})}\int_{\mathcal{X}\times \mathcal{Y}}c(x,y)d\pi(x,y),
\label{ot-primal-form}
\end{equation}

\vspace{-2mm}
where the minimum is taken over the transport plans $\pi$, i.e., the measures on $\mathcal{X}\times\mathcal{Y}$ whose marginals are $\mathbb{P}$ and $\mathbb{Q}$ (Figure \ref{fig:ot-plan-def}). The optimal $\pi^{*}\in\Pi(\mathbb{P},\mathbb{Q})$ is called the \textit{optimal transport plan}.

With mild assumptions on the transport cost $c(x,y)$ and distributions $\mathbb{P}$, $\mathbb{Q}$, the minimizer $\pi^{*}$ of \eqref{ot-primal-form} always exists \cite[Theorem 4.1]{villani2008optimal} but might not be unique. If $\pi^{*}$ is of the form $[\text{id} , T^{*}]_{\#} \mathbb{P}\in\Pi(\mathbb{P},\mathbb{Q})$ for some $T^*$, then
$T^{*}$ is an optimal transport map that minimizes \eqref{ot-primal-form-monge}.

\begin{figure}[!h]
\centering
\vspace{-1mm}
\begin{subfigure}[b]{0.48\linewidth}
  \begin{center}
    \includegraphics[width=0.85\linewidth]{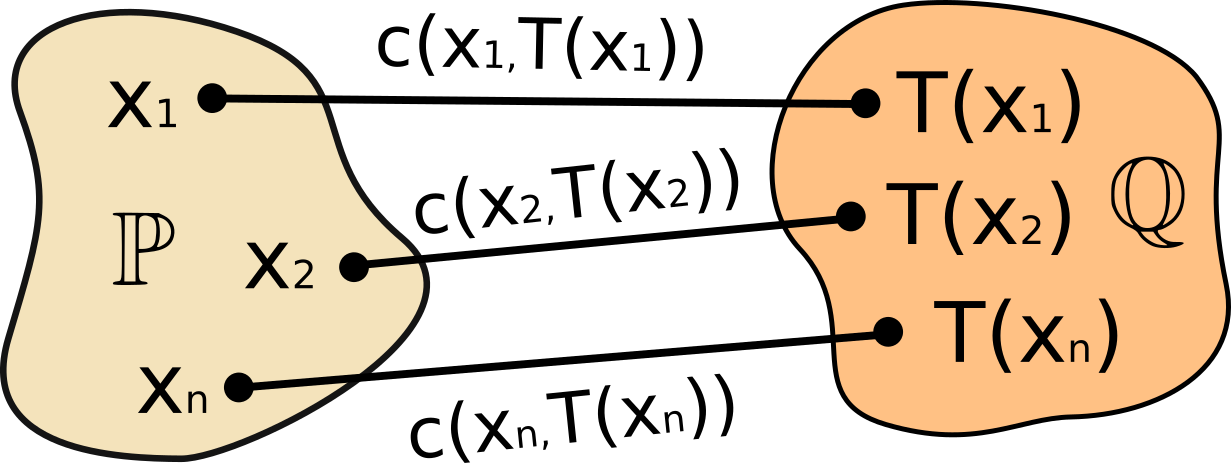}
  \end{center}
  \vspace{-1.7mm}
  \caption{\centering Monge's formulation of OT.}
  \label{fig:ot-map-def}
\end{subfigure}
\hspace{2mm}
\begin{subfigure}[b]{0.48\linewidth}
  \begin{center}
    \includegraphics[width=0.85\linewidth]{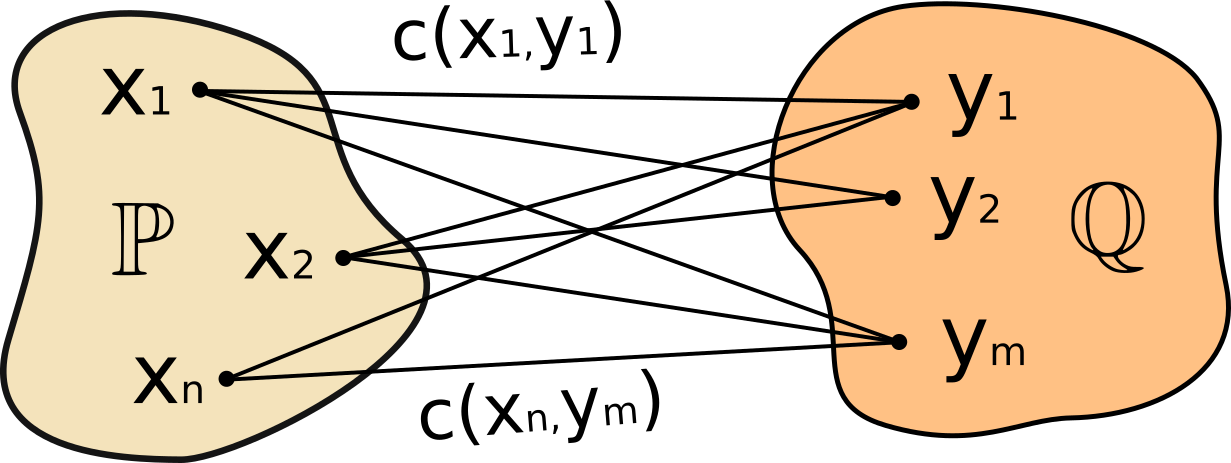}
  \end{center}
  \vspace{-1.7mm}
  \caption{\centering Kantorovich's formulation of OT.}
  \label{fig:ot-plan-def}
\end{subfigure}
\caption{Monge's and Kantorovich's formulations of Optimal Transport.}
\end{figure}

\vspace{-0mm}
\noindent\textbf{Dual form}. The dual form \cite{villani2003topics} of OT cost \eqref{ot-primal-form} is as follows:
\vspace{-1.7mm}
\begin{eqnarray}\text{Cost}(\mathbb{P}, \mathbb{Q})=
\sup_{f}\bigg[\int_{\mathcal{X}} f^{c}(x)d\mathbb{P}(x)+\int_{\mathcal{Y}} f(y)d\mathbb{Q}(y)\bigg]; 
\label{ot-dual-form-c}
\end{eqnarray}
\vspace{-3.8mm}

here $\sup$ is taken over all $f\!\in\! \mathcal{L}^{1}(\mathbb{Q})$, and ${f^{c}(x)\!=\!\inf\limits_{y\in\mathcal{Y}}\!\big[c(x,y)\!-\!f(y)\!\big]}$ is the $c$-transform of $f$.

\noindent\textbf{Optimal Transport in Generative Models.} The majority of existing OT-based generative models employ OT cost as the loss function to update the generative network, e.g., see \cite{arjovsky2017wasserstein}. These methods are out of scope of the present paper,  since they do not compute OT maps. Existing methods to compute the OT map approach the primal \eqref{ot-primal-form-monge}, \eqref{ot-primal-form} or dual form \eqref{ot-dual-form-c}. Primal-form methods \cite{lu2020large,xie2019scalable,bousquet2017optimal,balaji2020robust} optimize complex GAN objectives such as \eqref{base-gan-c} and provide biased solutions (\wasyparagraph\ref{sec-biased-ot}, \wasyparagraph\ref{sec-bias-experiments}). For a comprehensive overview of dual-form methods, we refer to \cite{korotin2021neural}. The authors conduct an evaluation of OT methods for the quadratic cost $c(x,y)=\|x-y\|^{2}$. According to them, the best performing method is $\lfloor \text{MM:R}\rceil$. 
Extensions of $\lfloor\text{MM:R}\rceil$ appear in \cite{rout2022generative,fan2023neural}.

\vspace{-2.5mm}
\section{Biased OT in GANs}
\vspace{-2mm}
\label{sec-biased-ot}
In this section, we establish connections between GAN methods regularized by content losses \eqref{identity-c-loss} and OT. Such GANs are popular in a variety of tasks beside SR, e.g., style transfer \cite{huang2018multimodal}. The theoretical analysis in this section holds for these tasks as well. However, since we empirically demonstrate the findings on the SR problem, we keep the corresponding notation in \wasyparagraph\ref{sec-biased-ot}. 

For a theoretical analysis, we stick to the basic formulation regularized with generic content loss \eqref{base-gan-c}. It represents the simplest and straightforward SR setup. We prove the following lemma, which connects the solution $T^{\lambda}$ of \eqref{base-gan-c} and OT maps.

\begin{wrapfigure}{r}{0.45\textwidth}
  \vspace{-5mm}
  \begin{center}
    \includegraphics[width=0.85\linewidth]{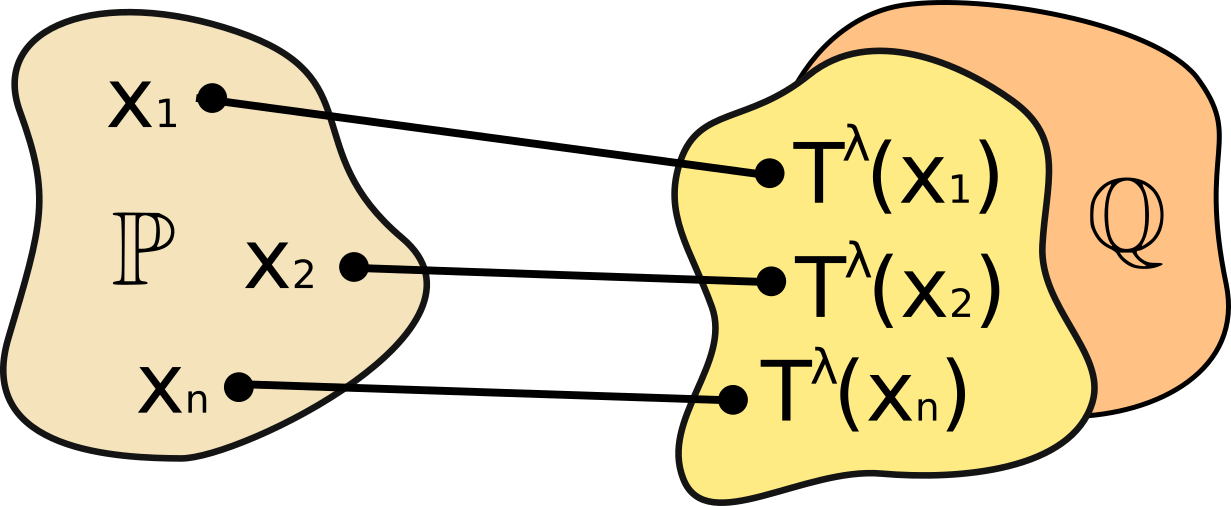}
  \end{center}
  \vspace{-4.0mm}
  \caption{\centering Illustration of Lemma \ref{lemma-optimal}. The solution $T^{\lambda}$ of \eqref{base-gan-c} is an OT map from $\mathbb{P}$ to $T^{\lambda}_{\#}\mathbb{P}$. In general, $T^{\lambda}_{\#}\mathbb{P}\neq \mathbb{Q}$ (Thm. \ref{lemma-biased}).}
  \label{fig:ot-map-biased}
\vspace{-4mm}
\end{wrapfigure}
\begin{lemma}[The solution of the regularized GAN is an OT map] Assume that $\lambda>0$ and the minimizer $T^{\lambda}$ of \eqref{base-gan-c} exists. Then $T^{\lambda}$ is an OT map between $\mathbb{P}$ and $\mathbb{Q}^{\lambda}\stackrel{\text{\normalfont def}}{=}T^{\lambda}_{\#}\mathbb{P}$ for cost $c(x,y)$, i.e., it minimizes
\vspace{-3mm}
\begin{equation*}
\inf_{T_{\#}\mathbb{P}=\mathbb{Q}^{\lambda}}\mathcal{R}_{c}(T)=\inf_{T_{\#}\mathbb{P}=\mathbb{Q}^{\lambda}}\int_{\mathcal{X}}c\big(x,T(x)\big)d\mathbb{P}(x).
\end{equation*}
\vspace{-4mm}
\label{lemma-optimal}
\end{lemma}
\vspace{-1mm}

Our Lemma \ref{lemma-optimal} states that the minimizer $T^{\lambda}$ of a regularized GAN problem is \textit{always} an OT map between $\mathbb{P}$ and the distribution $\mathbb{Q}^{\lambda}$ generated by the same $T^{\lambda}$ from $\mathbb{P}$. However, below we prove that $\mathbb{Q}^{\lambda}\neq\mathbb{Q}$, i.e., $T^{\lambda}$ \textbf{does not} actually produce the distribution of HR images (Figure \ref{fig:ot-map-biased}). To begin with, we prove the following auxiliary result.

\begin{lemma}[Reformulation of the regularized GAN via distributions]
Under the assumptions of Lemma \ref{lemma-optimal}, let $\mathcal{X}=\mathcal{Y}$ be a compact subset of $\mathbb{R}^D$ with negligible boundary. Let $\mathbb{P}\!\in\! \mathcal{P(X)}$ be absolutely continuous,  $\mathbb{Q}\!\in\!\mathcal{P(Y)}$ and $c(x,y)\!=\!\|x-y\|^{p}$ with $p>1$.
Then \eqref{base-gan-c} is equivalent to
\vspace{-2mm}
\begin{equation}
\inf_{\mathbb{Q}'\in\mathcal{P}(\mathcal{Y})}\mathcal{F}(\mathbb{Q}')\stackrel{\text{\normalfont{def}}}{=}\!\inf_{\mathbb{Q}'\in\mathcal{P}(\mathcal{Y})}\big[\mathcal{D}(\mathbb{Q}', \mathbb{Q})+\lambda\cdot\text{\normalfont{Cost}}(\mathbb{P}, \mathbb{Q}')\big], 
\label{alternative-gan-c}
\end{equation}
and the solutions of \eqref{base-gan-c} and \eqref{alternative-gan-c} are related as $\mathbb{Q}^{\lambda}=T^{\lambda}_{\#}\mathbb{P}$, where $\mathbb{Q}^{\lambda}$ is the minimizer of \eqref{alternative-gan-c}.
\label{corollary-equivalence}
\end{lemma}
\vspace{-1mm}

\vspace{-1mm}
In the following Theorem, we prove that, in general, 
$\mathbb{Q}^{\lambda}\neq\mathbb{Q}$ for the minimizer $\mathbb{Q}^{\lambda}$ of \eqref{alternative-gan-c}.

\begin{theorem}[The distribution solving the regularized GAN problem is always biased]
Under the assumptions of Lemma \ref{corollary-equivalence}, assume that the first variation \cite[\text{\normalfont{Definition 7.12}}]{santambrogio2015optimal} of the functional ${\mathbb{Q}'\mapsto\mathcal{D}(\mathbb{Q}',\mathbb{Q})}$ at the point $\mathbb{Q}'=\mathbb{Q}$ exists and is equal to zero. This means that $\mathcal{D}(\mathbb{Q}+\epsilon\Delta\mathbb{Q}, \mathbb{Q})=\mathcal{D}(\mathbb{Q}, \mathbb{Q})+ o(\epsilon)$ for every signed measure $\Delta\mathbb{Q}$ of zero total mass and $\epsilon\geq 0$ such that $\mathbb{Q}+\epsilon\Delta\mathbb{Q}\in\mathcal{P}(\mathcal{Y})$.  
Then, if $\mathbb{P}\neq\mathbb{Q}$, then  
$\mathbb{Q}'=\mathbb{Q}$ does not deliver the minimum to $\mathcal{F}$. 
\label{lemma-biased}
\end{theorem}
\vspace{-1mm}
Before proving Theorem \ref{lemma-biased}, we highlight that the assumption about the vanishing first variation of $\mathbb{Q}'\mapsto\mathcal{D}(\mathbb{Q}',\mathbb{Q})$ at $\mathbb{Q}'=\mathbb{Q}$ is \textit{reasonable}. In Appendix \ref{sec-first-variation}, we prove that this \underline{assumption holds} for the popular GAN discrepancies $\mathcal{D}(\mathbb{Q}',\mathbb{Q})$, e.g., $f$-divergences \cite{nowozin2016f} and certain Wasserstein distances \cite{arjovsky2017wasserstein}.
 \vspace{1mm}\newline
\hspace*{-3mm}\fbox{
    \parbox{\textwidth}{
    \vspace{-1mm}
        \begin{corollary}
    Under the assumptions of Theorem \ref{lemma-biased}, the solution $T^{\lambda}$ of regularized GAN \eqref{base-gan-c} is \underline{\normalfont biased}, i.e., it does not satisfy $T^{\lambda}_{\#}\mathbb{P}=\mathbb{Q}$ and {\normalfont does not transform LR images to true HR ones}.
    \end{corollary}
    \vspace{-1mm}
    }
}

\vspace{1mm}
Additionally, we provide a toy example that  further illustrates the issue with the bias.
\begin{example}
Consider $\mathcal{X}=\mathcal{Y}=\mathbb{R}^{1}$. Let $\mathbb{P}=\frac{1}{2}\delta_{0}+ \frac{1}{2}\delta_{2}$, $\mathbb{Q}=\frac{1}{2}\delta_{1}+ \frac{1}{2}\delta_{3}$ be distributions concentrated at $\{0,2\}$ and $\{1,3\}$, respectively. Put $c(x,y)=|x-y|$ to be the content loss. Also, let $\mathcal{D}$ to be the OT cost for $|x-y|^2$. Then for $\lambda=0$ there exist two maps between $\mathbb{P}$ and $\mathbb{Q}$ that  deliver the same minimal value for \eqref{base-gan-c}, namely $T(0)=1,T(2)=3$ and $T(0)=3,T(2)=1$. For $\lambda>0$, the optimal solution of the problem \eqref{base-gan-c} is unique, \textbf{biased} and given by $T(0)=1-\frac{\lambda}{2},T(2)=3-\frac{\lambda}{2}$.
\label{example-toy-w2}
\end{example}

In Example \ref{example-toy-w2},\enskip    $T_{\#}^{\lambda}\mathbb{P}=\mathbb{Q}^{\lambda}$ \textit{never} matches $\mathbb{Q}$ exactly for $\lambda>0$. 
In \wasyparagraph\ref{sec-bias-experiments}, we conduct an evaluation of maps obtained via minimizing objective \eqref{base-gan-c} on the synthetic benchmark by \cite{korotin2021neural}. We empirically demonstrate that the bias exists and it is indeed a notable practical issue.
 
\noindent \textbf{Remarks.} Throughout this section, we enforce additional assumptions on \eqref{base-gan-c}, e.g., we restrict our analysis to content losses $c(\cdot,\cdot)$, which are powers of Euclidean norms $\|\cdot\|^{p}$. This is needed to make the derivations concise and to be able to exploit the available results in OT. We think that the provided results hold under more general assumptions and leave this question open for future studies.

\vspace{-1.5mm}\section{Relation between GANs and Neural Optimal Transport Solvers}
\label{sec-unbiased-ots}

\vspace{-1mm}
In this section, we analyze recent neural algorithms to compute OT maps \cite{fan2023neural,korotin2023neural,rout2022generative} and show their connection with regularized GANs. 
Below we show that their loss can be viewed as a particular (in a certain sense) GAN objective regularized with the content loss. To begin with, we recall that typical OT optimization objective is minimax and given by
\begin{eqnarray}
    [\text{Cost}(\mathbb{P},\mathbb{Q})=]\nonumber\\\!\!\!\!\!\!\qquad\!\!\!\!\sup_f\inf_{T:\mathcal{X}\mapsto\mathcal{Y}} \big[\int_{\mathcal{Y}} f(y)d\mathbb{Q}(y)+\int_{\mathcal{X}}\left\lbrace c\big(x,T(x)\big)\!\!-\!\!f(T(x))\right\rbrace d\mathbb{P}(x)\big],
    \label{ots-objective}
\end{eqnarray}
where $\sup_{f}$ is taken w.r.t. all potentials $f\in\mathcal{L}^{1}(\mathbb{Q})$. Under mild assumptions\footnote{In certain cases among the solutions in such a problem may be so-called \textit{fake} solutions which are not the OT maps. We refer to \cite{korotin2023kernel} for a fruitful discussion of this phenomena.}, by solving \eqref{ots-objective} one may recover the true (\textit{unbiased}) OT map $T^{*}$, see \cite[Lemma 4]{korotin2023neural}, \cite[\wasyparagraph 3,4]{fan2023neural}. In practice, $T,f$ are replaced with neural networks; as in GANs, they are optimized with the stochastic gradient descent-ascent techniques using the empirical samples from $\mathbb{P},\mathbb{Q}$.

Now let us get back to GANs. In \wasyparagraph\ref{sec-biased-ot}, we show that solutions of \eqref{base-gan-c} are, in general, \textit{biased} OT maps. Note, that this bias is related to the trade-off between components of GANs optimization objective \eqref{base-gan-c}, i.e., the quality of generated image and its similarity to the input. In order to resolve the bias issue, one can consider the loss $\mathcal{D}(T_{\#}\mathbb{P},\mathbb{Q})\equiv \mathcal{I}(T_{\#}\mathbb{P},\mathbb{Q})$ where $\mathcal{I}$ is the indicator function which takes two values:  zero if its inputs coincide and $+\infty$ when they differ. Then we can rewrite \eqref{base-gan-c} as
\begin{eqnarray}
\lambda \cdot \inf_{T:\mathcal{X}\mapsto\mathcal{Y}}\big[\frac{1}{\lambda}\mathcal{I}(T_{\#}\mathbb{P},\mathbb{Q})+\mathcal{R}_{c}(T)\big]
=\lambda \cdot\inf_{T:\mathcal{X}\mapsto\mathcal{Y}}\big[\mathcal{I}(T_{\#}\mathbb{P},\mathbb{Q})+\mathcal{R}_{c}(T)\big].
\label{gan-idt-loss}
\end{eqnarray}
Here we used the fact that $\lambda\cdot \mathcal{I}(\cdot,\cdot)=\mathcal{I}(\cdot,\cdot)$.
Note that the solution $\widehat{T}$ (if it exists)  of \eqref{gan-idt-loss} satisfies $\widehat{T}_{\#}\mathbb{P}=\mathbb{Q}$. Otherwise, the objective yields the value $+\infty$. Therefore, problem \eqref{gan-idt-loss} is equivalent to the optimization of the functional $\mathcal{R}_{c}(T)$ with the constraint $T_{\#}\mathbb{P}=\mathbb{Q}$. As a result, \eqref{gan-idt-loss} turns to be just the Monge OT problem \eqref{ot-primal-form-monge} multiplied by $\lambda>0$ and with the constraint incorporated directly to the loss via the indicator function $\mathcal{I}$. We conclude that its solution is an OT map, i.e., $\widehat{T}=T^{*}$, and the optimal value of \eqref{gan-idt-loss} is exactly $\lambda\cdot \text{Cost}(T_{\#}\mathbb{P},\mathbb{Q})$.

Unfortunately, optimizing objective \eqref{gan-idt-loss} in practice is non-trivial: even testing the condition $T_{\#}\mathbb{P}=\mathbb{Q}$ (i.e., computing $\mathcal{I}$) is hard, which makes it challenging to compute the loss. Note that
\begin{eqnarray}
\mathcal{I}(T_{\#}\mathbb{P},\mathbb{Q})=\sup_f \big[-\int_{\mathcal{X}} f(T(x))d\mathbb{P}(x) + \int_{ \mathcal{Y}} f(y)d\mathbb{Q}(y)\big],
\label{indicator-expression}
\end{eqnarray}
\vspace{-2mm}\newline where $f$ skims through all integrable w.r.t. $\mathbb{Q}$ and $T_{\#}\mathbb{P}$ functions. Indeed, if $T_{\#}\mathbb{P}=\mathbb{Q}$, the two integrals always coincide. Otherwise, there always exists a measurable function $f$ whose integrals over distributions differ. One may then multiply it by an arbitrary number to get any value of the expression, i.e., in this case, $\sup$ equals $+\infty$. We substitute \eqref{indicator-expression} to \eqref{gan-idt-loss} multiplied by $\frac{1}{\lambda}$ and get
\begin{eqnarray}
    \text{Cost}(\mathbb{P}, \mathbb{Q})=\inf_{T:\mathcal{X}\mapsto\mathcal{Y}}\mathcal{I}(T_{\#}\mathbb{P},\mathbb{Q})+\mathcal{R}_{c}(T) =\nonumber \\ \inf_{T:\mathcal{X}\mapsto\mathcal{Y}}\sup_f \big[\int_{y \in \mathcal{Y}} f(y)d\mathbb{Q}(y)+\int_{x\in\mathcal{X}}\left\lbrace c\big(x,T(x)\big)-f(T(x))\right\rbrace d\mathbb{P}(x)\big]
\label{gan-f-loss}
\end{eqnarray}
which almost coincides with \eqref{ots-objective}; the only difference is the order of $\inf$ and $\sup$. At this point, a natural question arises: what is the conceptual difference between \eqref{ots-objective} and \eqref{gan-f-loss}, and why neural OT works typically consider \eqref{ots-objective} rather than $\eqref{gan-f-loss}$? We believe that this is simply because the loss for the Neural OT methods is usually derived from the conventional dual formulation of OT \eqref{ot-dual-form-c} by expressing the $c$-transform, which yields the additional inner problem. In fact, when it comes to the practical optimization of \eqref{ots-objective} or \eqref{gan-f-loss}, the actual order of optimization does not matter too much. The overall performance depends more on a proper choice of hyperparameters of the optimization.

\vspace{-1.3mm}
\subsection{Regularized GANs vs. Optimal Transport Solver}
\vspace{-1.5mm}
\label{sec-gans-vs-ot}

In this subsection, we discuss similarities and differences between neural OT optimization objective \eqref{ots-objective} and the objective of regularized GANs \eqref{base-gan-c}. We establish an intriguing connection between GANs that use \textit{integral probability metrics} (IPMs) as $\mathcal{D}$. A~discrepancy $\mathcal{D}\!:\!\mathcal{P}(\mathcal{Y})\!\times\! \mathcal{P}(\mathcal{Y})\!\rightarrow\!\mathbb{R}_{+}$ is an IPM if
\vspace{-1mm}
\begin{equation}
    \hspace{-0.5mm}\mathcal{D}(\mathbb{Q}_{1},\mathbb{Q}_{2})\!=\!\sup_{f\in\mathcal{F}}\!\big[\!\int_{\mathcal{Y}}\!f(y)d\mathbb{Q}_{2}(y)\!-\!\int_{\mathcal{Y}}\!f(y)d\mathbb{Q}_{1}(y)\big]\!,
    \label{IPM}
\end{equation}
where the maximization is performed over some certain  class~$\mathcal{F}$ of functions (discriminators) ${f:\mathcal{Y}\rightarrow\mathbb{R}}$. The most popular example of $\mathcal{D}$ is the Wasserstein-1 loss \cite{arjovsky2017towards}, where $\mathcal{F}$ is a class of $1$-Lipschitz functions. For other IPMs, see \cite[Table 1]{mroueh2017sobolev}.

Substituting \eqref{IPM} to \eqref{base-gan-c} yields the saddle-point optimization problem for the \textbf{regularized IPM GAN}:\vspace{-1mm}
\begin{eqnarray}
\inf_{T:\mathcal{X}\rightarrow\mathcal{Y}}\bigg[\sup_{f\in\mathcal{F}}\left\lbrace\!\int_{\mathcal{Y}}\!f(y)d\mathbb{Q}(y)\!-\!\int_{\mathcal{X}}\!f\big(T(x)\big)d\mathbb{P}(x)\right\rbrace+
\lambda\int_{\mathcal{X}}c\big(x,T(x)\big)d\mathbb{P}(x)\rbrace\bigg]
\nonumber
=\\ \inf_{T:\mathcal{X}\rightarrow\mathcal{Y}}\sup_{f\in\mathcal{F}}\bigg[\int_{\mathcal{Y}} f(y)d\mathbb{Q}(y)+
\int_{\mathcal{X}}\left\lbrace  \lambda\cdot c\big(x,T(x)\big)-f\big(T(x)\big)\right\rbrace d\mathbb{P}(x)\bigg].\quad
\label{main-loss-gan}
\vspace{-4.0mm}
\end{eqnarray}
We emphasize that the expression inside \eqref{main-loss-gan} for $\lambda=1$ is similar to the expression in OTS optimization \eqref{ots-objective}. Below we highlight the \textbf{key differences} between \eqref{ots-objective} and \eqref{main-loss-gan}.  

\textbf{First}, in OTS the optimization over potential $f$ is unconstrained, while in IPM GAN it must belong to $\mathcal{F}$, some certain restricted class of functions. For example, when ${\mathcal{D}}$ is the Wasserstein-1 ($\mathbb{W}_{1}$) IPM, one has to use an additional penalization, e.g., the gradient penalty \cite{gulrajani2017improved}. This further complicates the optimization and adds hyperparameters which have to  be carefully selected. 

\textbf{Second}, the optimization of IPM GAN requires selecting a parameter $\lambda$ that balances the content loss $\mathcal{R}_{c}$ and the discrepancy $\mathcal{D}$. In OTS for all costs $\lambda\cdot c(x,y)$ with $\lambda>0$, the OT map $T^{*}$ is the same. 
\vspace{-0.3mm}

To conclude, even for $\lambda=1$, the IPM GAN problem generally does not match that of OTS. Table \ref{table-comparison-gan} summarizes the differences and the similarities between OTS and regularized IPM GANs.

\begin{table*}[!t]
\vspace{-2mm}
\centering
\scriptsize
\hspace*{-1.5mm}\begin{tabular}{c|c|c}
\toprule
 & \textbf{Optimal Transport Solver (OTS)} & \textbf{Regularized IPM GAN} \hspace{-2mm} \\ \midrule
\makecell{Minimax\\optimization\\objective} & \makecell{$\sup\limits_{f}\inf\limits_{T:\mathcal{X}\rightarrow\mathcal{Y}}\bigg[\int_{\mathcal{Y}} f(y)d\mathbb{Q}(y)+$\\
$\int_{\mathcal{X}}\left\lbrace c\big(x,T(x)\big)-f\big(T(x)\big)\right\rbrace d\mathbb{P}(x)\bigg]$} & 
\makecell{$\inf\limits_{T:\mathcal{X}\rightarrow\mathcal{Y}}\sup\limits_{f\in\mathcal{F}}\bigg[\int_{\mathcal{Y}} f(y)d\mathbb{Q}(y)+$\\
$\int_{\mathcal{X}}\left\lbrace  \lambda\cdot c\big(x,T(x)\big)-f\big(T(x)\big)\right\rbrace d\mathbb{P}(x)\bigg]$} \hspace{-2mm}
\\ \midrule
\makecell{Potential $f$\\(discriminator)} & {\color{LimeGreen}Unconstrained} $f\in L^{1}(\mathbb{Q})$ & \makecell{{\color{red}Constrained} $f\in \mathcal{F}\subset L^{1}(\mathbb{Q})$\\
A method to impose \\the constraint is needed.
} \hspace{-2mm} \\ \midrule
\makecell{Regularization\\weight $\lambda$} & {\color{LimeGreen}N/A} & {\color{red}Hyperparameter} choice required \hspace{-1.5mm}\\
\bottomrule
\end{tabular}
\vspace{-1.5mm}
\caption{\centering Comparison of the optimization objectives of OTS and regularized IPM GAN.}
\label{table-comparison-gan}
\vspace{-3mm}
\end{table*}

\vspace{-5mm}\section{Experimental Illustration}
\vspace{-2mm}
\label{sec-experiments}
In \wasyparagraph\ref{sec-bias-experiments}, we assess the bias of regularized IPM GANs by using the publicly available Wasserstein-2 benchmark\footnote{\url{https://github.com/iamalexkorotin/Wasserstein2Benchmark}} \cite{korotin2021neural}. In \wasyparagraph\ref{sec-experiments-aim}, we evaluate OTS on the large-scale unpaired AIM-19 dataset from \cite{lugmayr2019aim} and compare it with popular GAN-based solutions for unpaired image SR. The data is publicly available at \url{https://competitions.codalab.org/competitions/20164}.
The code is written in \texttt{PyTorch} and is available at 
\begin{center}\url{https://github.com/milenagazdieva/OT-Super-Resolution}.
\end{center}
The hyperparameters for Algorithm \ref{algorithm-ot} are listed in Table \ref{table-params} of Appendix \ref{sec-details}.
 
\noindent\textbf{Neural network architectures.} We use WGAN-QC's \cite{liu2019wasserstein} ResNet \cite{he2016deep} architecture for the potential $f_{\omega}$. In \wasyparagraph\ref{sec-bias-experiments}, where input and output images have the same size, we use UNet\footnote{\url{github.com/milesial/Pytorch-UNet}} \cite{ronneberger2015u} as a transport map $T_{\theta}$. In \wasyparagraph\ref{sec-experiments-aim}, the LR input images are $4\times 4$ times smaller than HR, so we use 
EDSR network \cite{lim2017enhanced}.

\noindent\textbf{Transport costs.} In \wasyparagraph\ref{sec-bias-experiments}, 
we use the \textit{mean squared error} (MSE), i.e.,\newline  ${c(x,y)=\|x-y\|^{2}/\dim (\mathcal{Y})}$. It is equivalent to the quadratic cost but is more convenient due to the normalization. In \wasyparagraph\ref{sec-experiments-aim}, 
we consider $c(x,y)=b(\text{Up}(x),y)$, where $b$ is a cost between the bicubically upsampled LR image $x^{\text{up}}= \text{Up}(x)$ and HR image $y$. We test $b$ defined as $\text{MSE}$ and the \underline{\textit{perceptual cost}} using features of a pre-trained VGG-16 network \cite{simonyan2014very}, see Appendix \ref{sec-details} for details.

\vspace{-4.5mm}
\subsection{Assessing the Bias in Regularized GANs}
\vspace{-2.5mm}
\label{sec-bias-experiments}

In this section, we empirically confirm the insight of \wasyparagraph\ref{sec-biased-ot} that the solution $T^{\lambda}$ of \eqref{base-gan-c} may not satisfy $T^{\lambda}_{\#}\mathbb{P}=\mathbb{Q}$. Notably, if $T^{\lambda}_{\#}\mathbb{P}=\mathbb{Q}$, then from our Lemma \ref{lemma-optimal} it follows that $T^{\lambda}\equiv T^{*}$, where $T^{*}$ is an OT map from $\mathbb{P}$ to $\mathbb{Q}$ for $c(x,y)$. Thus, to access the bias, it is reasonable to compare the learned map $T^{\lambda}$ with the ground truth OT map $T^{*}$ for $\mathbb{P}$, $\mathbb{Q}$.

For evaluation, we use the Wasserstein-2 benchmark \cite{korotin2021neural}. It provides high-dimensional continuous pairs $\mathbb{P}$, $\mathbb{Q}$ with an \underline{\textit{analytically known}} OT map $T^{*}$ for the quadratic cost ${c(x,y)=\|x-y\|^{2}}$. We use their ``Early" images benchmark pair. It simulates the image deblurring setup, i.e., $\mathcal{X}=\mathcal{Y}$ is the space of $64\times 64$ RGB images, $\mathbb{P}$ is blurry faces, $\mathbb{Q}$ is clean faces satisfying $\mathbb{Q}=T^{*}_{\#}\mathbb{P}$, where $T^{*}$ is an analytically known OT map, see the 1st and 2nd lines in Figure \ref{fig:ipm-vs-ots-benchmark}.

\begin{figure}
  \vspace{-2mm}\begin{center}
    \includegraphics[width=0.95\linewidth]{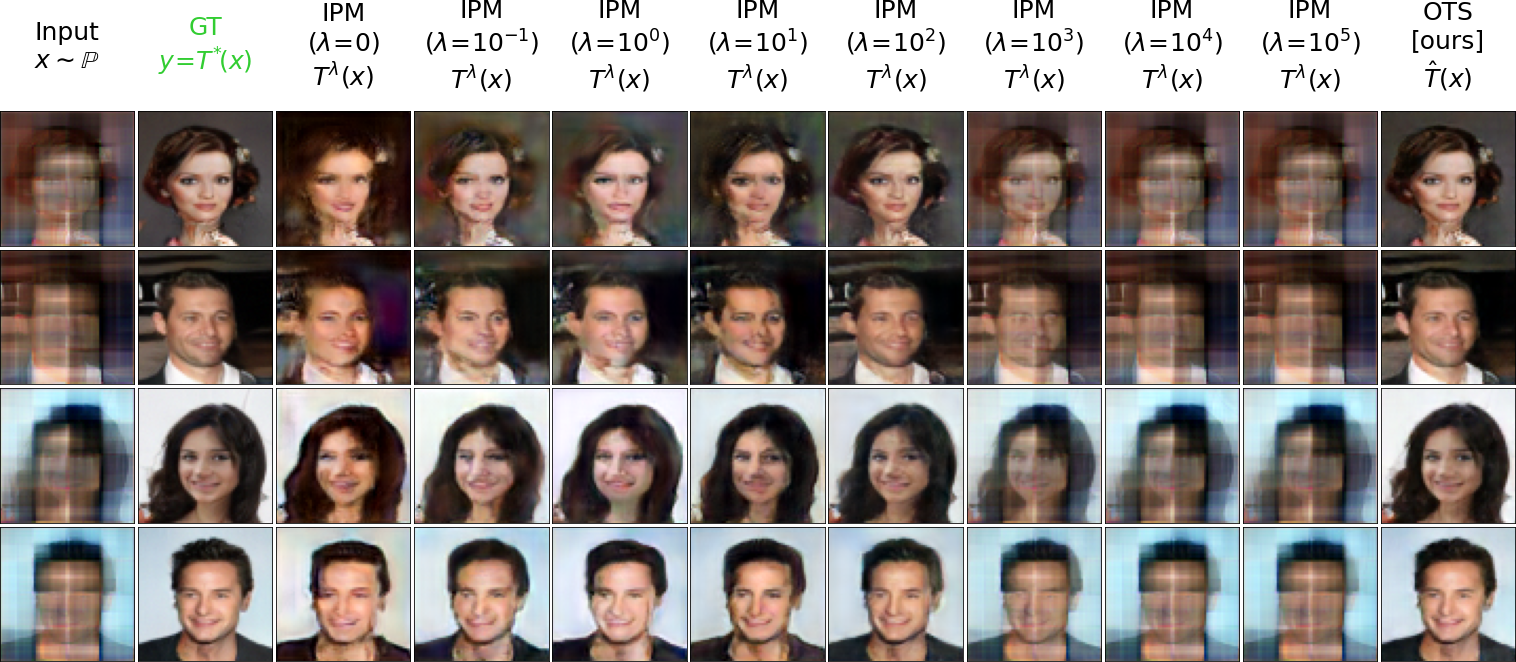}
  \end{center}
  \vspace{-3.5mm}
  \caption{\centering Comparison of OTS, regularized IPM GAN on the Wasserstein-2 benchmark. The 1st line shows blurry faces $x\sim\mathbb{P}$, the 2nd line, clean faces $y=T^{*}(x)$, where $T^{*}$ is the OT map from $\mathbb{P}$ to $\mathbb{Q}$. Next lines show maps from $\mathbb{P}$ to $\mathbb{Q}$ 
  fitted by the methods.}
  \label{fig:ipm-vs-ots-benchmark}
\end{figure}
\begin{table}[t!]
\vspace{-2.5mm}
\centering
\scriptsize
\begin{tabular}{@{\hskip1.9pt}c@{\hskip1.9pt}|@{\hskip1.9pt}c@{\hskip1.9pt}|@{\hskip1.9pt}c@{\hskip1.9pt}|@{\hskip1.9pt}c@{\hskip1.9pt}|@{\hskip1.9pt}c@{\hskip1.9pt}|@{\hskip1.9pt}c@{\hskip1.9pt}|@{\hskip1.9pt}c@{\hskip1.9pt}|@{\hskip1.9pt}c@{\hskip1.9pt}|@{\hskip1.9pt}c@{\hskip1.9pt}|@{\hskip1.9pt}c@{\hskip1.9pt}}\toprule
\multirow{2}{*}{\makecell{\textit{Metrics}/ \\ \textit{Method}}} & \multicolumn{8}{c|}{\textbf{Regularized} \textbf{IPM GAN} (WGAN-GP, $\lambda_{\text{GP}}=10$)} & \multirow{2}{*}{\makecell{\textbf{OTS}}} \\
\cline{2-9}
& $\lambda=0$ & $\lambda=10^{-1}$ & $\lambda=10^{0}$ & $\lambda=10^{1}$ & $\lambda=10^{2}$ & $\lambda=10^{3}$ & $\lambda=10^{4}$ & $\lambda=10^{5}$ & \\
\midrule
$\mathcal{L}^{2}\text{-UVP}\downarrow$ & $25.2\%$ & $16.7\%$ & $17.7\%$ & $12.0\%$ & $\textbf{4.0}\%$ & $14.0\%$ & $28.5\%$ & $30.5\%$ & $\textbf{1.4}\%$\\
\textit{FID}$\downarrow$ & $57.24$ & $46.23$ & $40.04$ & $42.89$ & $\textbf{24.25}$ & $187.95$ & $332.7$ & $334.7$ & $\mathbf{15.65}$ \\
\textit{PSNR}$\uparrow$ & $17.90$ & $19.76$ & $19.34$ & $20.81$ & $\textbf{25.58}$ & $19.91$ & $16.90$ & $16.52$ & $\mathbf{30.02}$ \\
\textit{SSIM}$\uparrow$ & $0.565$ & $0.655$ & $0.656$ & $0.689$ & $\textbf{0.859}$ & $0.702$ & $0.520$ & $0.498$ & $\mathbf{0.933}$\\
\textit{LPIPS}$\downarrow$ & $0.135$ & $0.093$ & $0.099$ & $0.081$ & $\textbf{0.031}$ & $0.172$ & $0.429$ & $0.446$ & $\mathbf{0.013}$ \\
\bottomrule
\end{tabular}
\vspace{-2mm}
\caption{\centering\protect{Quantitative evaluation of restoration maps fitted by the regularized IPM GAN, OTS using the Wasserstein-2 images benchmark \cite{korotin2021neural}}.}
\label{table-benchmark}
\vspace{-4mm}
\end{table}

To quantify the learned maps from $\mathbb{P}$ to $\mathbb{Q}$, we use PSNR, SSIM, LPIPS \cite{zhang2018perceptual}, FID \cite{heusel2017gans} metrics. Similar to \cite{Wei_2021_CVPR}, we use the AlexNet-based \cite{krizhevsky2012imagenet} LPIPS. FID and LPIPS are practically the \textit{most important} since they better correlate with the human perception of the image quality. We include PSNR, SSIM as popular evaluation metrics, but they are known to \textit{badly measure perceptual quality} \cite{zhang2018perceptual,nilsson2020understanding}. Due to this, higher PSNR, SSIM values do not necessarily mean better performance. We provide additional \underline{\textit{details}} on these metrics in Appendix \ref{sec-details-eval}. In this section, we additionally use the $\mathcal{L}^{2}\text{-UVP}$ \cite[\wasyparagraph 4.2]{korotin2021neural} metric. 

\vspace{-0.5mm}
On the benchmark, we compare OTS \eqref{ots-objective} and IPM GAN \eqref{base-gan-c}. We use MSE as the content loss $c(x,y)$. In IPM GAN, we use the Wasserstein-1 ($\mathbb{W}_{1}$) loss with the gradient penalty ${\lambda_{\text{GP}}=10}$ \cite{gulrajani2017improved} as $\mathcal{D}$. We do $10$ discriminator updates per $1$ generator update and train the model for 15K generator updates. For fair comparison, the rest hyperparameters match those of OTS algorithm. We train the regularized WGAN-GP with various coefficients of content loss $\lambda\in \{0, 10^{-1}, \dots, 10^{5}\}$ and show the learned maps $T^{\lambda}$ and the map $\hat{T}$ obtained by OTS in Figure \ref{fig:ipm-vs-ots-benchmark}.

\noindent\textbf{Results.} The performance of the regularized IPM GAN \underline{\textit{significantly}} depends on the choice of the content loss value $\lambda$. For high values $\lambda\geq 10^{3}$, the learned map is close to the identity as expected. For small values $\lambda\leq 10^{1}$, the regularization has little effect, and WGAN-GP solely struggles to fit a good restoration map. Even for the best performing $\lambda=10^{2}$ all metrics are notably worse than for OTS. Importantly, \underline{\textit{OTS decreases the burden of parameter searching}} as there is no parameter $\lambda$. 

\vspace{-5.1mm}
\subsection{Large-Scale Evaluation}
\vspace{-2.3mm}
\label{sec-experiments-aim}
For evaluating OTS method at a large-scale, we employ the dataset by \cite{lugmayr2019aim} of AIM 2019 Real-World Super-Resolution Challenge (Track 2). The train part contains 800 HR images with up to 2040 pixels width or height and 2650 unpaired LR images of the same shape. They are constructed using artificial, but realistic, image degradations. We quantitatively evaluate OTS method on the validation part of AIM dataset that contains 100 pairs of LR-HR images.

\begin{figure*}[!t]
\centering
\vspace{-3mm}
\includegraphics[width=0.8\linewidth]{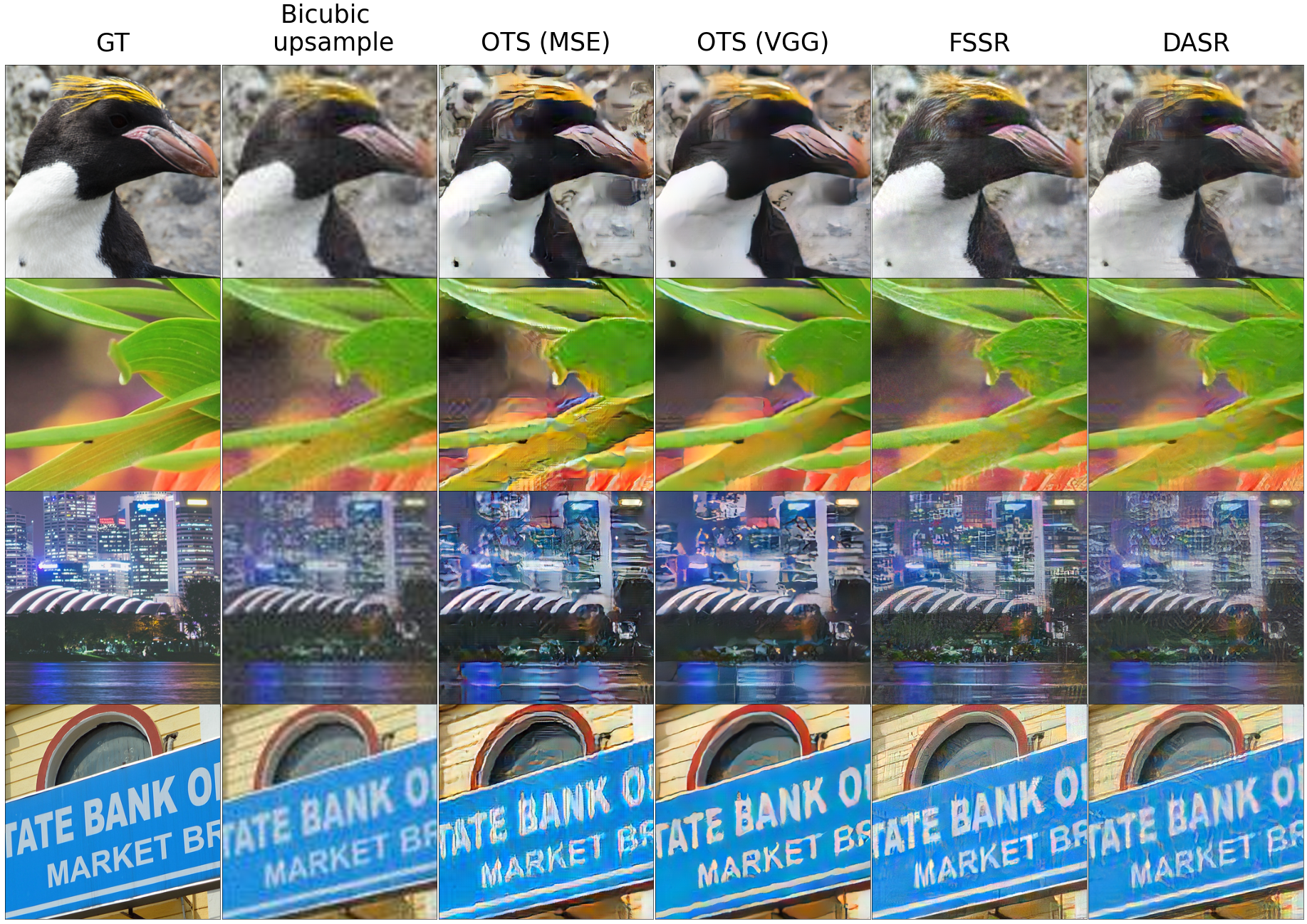}
\vspace{-2mm}
\caption{\centering Qualitative results of OTS, bicubic upsample, FSSR and DASR \protect\linebreak on AIM 2019 dataset (350$\times$350 crops).}
\label{fig:aim-results}
\vspace{-7mm}
\end{figure*}

\noindent\textbf{Baselines}.
We compare OTS on AIM dataset with the bicubic upsample, FSSR \cite{fritsche2019frequency} and DASR \cite{Wei_2021_CVPR} methods. FSSR method is the winner of AIM 2019 Challenge; DASR is another competetive method for unpaired image SR. Both methods utilize the idea of frequency separation and solve the problem in two steps. First, they train a network to generate LR images. Next, they train a super-resolution network using generated pseudo-pairs. Differently to FSSR, DASR also employs real-world LR images for training SR network taking into consideration the domain gap between generated and real-world LR images. Both methods utilize several losses, e.g., adversarial and perceptual, either on the entire image or on its high/low frequency components. 
For testing FSSR and DASR, we use their official code and pretrained models.

\looseness=-1
\noindent \textbf{Implementation details.}
We train the networks using 128$\times$128 HR, 32$\times$32 LR random \textit{patches} of images augmented via random flips, rotations. We conduct separate experiments using EDSR as the transport map and either MSE or perceptual cost, and denote them as OTS (MSE), OTS (VGG) respectively. 

\noindent\textbf{Metrics.} We calculate PSNR, SSIM, LPIPS, FID. FID is computed on $32\hspace{-0.7mm}\times\hspace{-0.7mm}32$ patches of LR test images upsampled by the method in view w.r.t.\ random patches of test HR.
We use 50k patches to compute FID. 
The other metrics are computed on the \textit{entire} upsampled LR test and HR test images.

\begin{wraptable}{r}{8.3cm}
\vspace{-4.5mm}
\centering
\scriptsize
\begin{tabular}{c| c c c c}
\toprule
 \textbf{Method} & \textbf{FID} $\downarrow$ & \textbf{PSNR} $\uparrow$ & \textbf{SSIM} $\uparrow$ & \textbf{LPIPS} $\downarrow$\\ \midrule
 \makecell{Bicubic upsample} &  178.59 & {\color{LimeGreen}22.39} &  {\color{LimeGreen}0.613} & 0.688\\ \midrule
 \makecell{OTS (MSE)} & 139.17 & 19.73 &  0.533 & 0.456\\  \midrule
 \makecell{OTS (VGG)} & {\color{blue}89.04} & \underline{20.96} &  {\color{blue}0.605} & {\color{blue}0.380}\\  \midrule
 FSSR & {\color{LimeGreen}53.92} & 20.83 &  0.514 & \underline{0.390}\\ \midrule
 DASR & \underline{124.09} & {\color{blue}21.79} &  \underline{0.577} & {\color{LimeGreen}0.346}\\ \bottomrule
 \end{tabular}
 \vspace{-2.8mm}
\caption{\centering \small Comparison of OTS with FSSR, DASR on AIM19 dataset. The 1st, 2nd, 3rd best results are highlighted in {\color{LimeGreen}green}, {\color{blue}blue} and \underline{underlined}, respectively.}
\label{table-aim}
\vspace{-5.3mm}
\end{wraptable}
\looseness=-1
\noindent\textbf{Experimental results} are given in 
Table \ref{table-aim}, Figure \ref{fig:aim-results}. The table shows that results across different metrics are not always consistent, but one key takeaway is that incorporating a perceptual cost function in OTS significantly enhances its performance. According to perceptual metrics (FID and LPIPS), OTS with a perceptual cost outperforms DASR and is comparable to FSSR. Specifically, OTS (VGG) surpasses DASR in FID, while its LPIPS score is only slightly worse, reinforcing its strong perceptual quality. Moreover, OTS (VGG) outperforms FSSR in PSNR, SSIM, and, more importantly, LPIPS, although it is slightly worse in FID. Notably, bicubic upsampling achieves the highest PSNR and SSIM scores, yet its visual quality remains inferior, further confirming the limitations of these metrics (see \wasyparagraph 6.1, Appendix \ref{sec-details-eval}). Qualitative analysis further supports these findings, showing that OTS with a perceptual cost function better deals with noise artifacts.
\underline{Additional results} are given in Appendix \ref{sec-additional-results}. We also demonstrate the \underline{bias issue} of FSSR and DASR in Appendix \ref{sec-aim-bias}.
Overall, our evaluation indicates that OTS is a promising approach for unpaired image SR, demonstrating applicability to large-scale experiments where it achieves competitive performance, while also helping to mitigate the bias issue common in GAN-based methods.

\vspace{-3mm}
\section{Conclusion}
\label{sec-discussion}
\vspace{-3.7mm}
Our analysis connects content losses in GANs with OT and 
reveals the bias issue. Content losses are used in a wide range of tasks besides SR, e.g., in the style transfer and domain adaptation tasks \cite{shrivastava2017learning,taigman2017unsupervised,zhu2017unpaired}. Our results demonstrate that GAN-based methods in all these tasks may \textit{a priori lead to biased solutions}. In certain cases it is undesirable, e.g., in medical applications \cite{bissoto2021gan}. Failing to learn true data statistics (and learning biased ones instead), e.g., in the super-resolution of MRI images, might lead to a wrong diagnosis made by a doctor due to SR algorithm drawing \textbf{inexistent details} on the scan. Thus, we think it is essential to emphasize and alleviate the bias issue.

\begin{acknowledgements}
    The work was supported by the grant for research centers in the field of AI provided by the Ministry of Economic Development of the Russian Federation in accordance with the agreement 000000C313925P4F0002 and the agreement with Skoltech №139-10-2025-033.
\end{acknowledgements}

\bibliographystyle{spmpsci.bst} 
\bibliography{references}

\begin{thebibliography}{10}
\providecommand{\url}[1]{{#1}}
\providecommand{\urlprefix}{URL }
\expandafter\ifx\csname urlstyle\endcsname\relax
  \providecommand{\doi}[1]{DOI~\discretionary{}{}{}#1}\else
  \providecommand{\doi}{DOI~\discretionary{}{}{}\begingroup \urlstyle{rm}\Url}\fi

\bibitem{arjovsky2017towards}
Arjovsky, M., Bottou, L.: Towards principled methods for training generative adversarial networks.
\newblock arXiv preprint arXiv:1701.04862  (2017)

\bibitem{arjovsky2017wasserstein}
Arjovsky, M., Chintala, S., Bottou, L.: Wasserstein generative adversarial networks.
\newblock In: International conference on machine learning, pp. 214--223. PMLR (2017)

\bibitem{balaji2020robust}
Balaji, Y., Chellappa, R., Feizi, S.: Robust optimal transport with applications in generative modeling and domain adaptation.
\newblock Advances in Neural Information Processing Systems \textbf{33}, 12,934--12,944 (2020)

\bibitem{bissoto2021gan}
Bissoto, A., Valle, E., Avila, S.: Gan-based data augmentation and anonymization for skin-lesion analysis: A critical review.
\newblock In: Proceedings of the IEEE/CVF Conference on Computer Vision and Pattern Recognition, pp. 1847--1856 (2021)

\bibitem{borji2019pros}
Borji, A.: Pros and cons of gan evaluation measures.
\newblock Computer vision and image understanding \textbf{179}, 41--65 (2019)

\bibitem{bousquet2017optimal}
Bousquet, O., Gelly, S., Tolstikhin, I., Simon-Gabriel, C.J., Schoelkopf, B.: From optimal transport to generative modeling: the vegan cookbook.
\newblock arXiv preprint arXiv:1705.07642  (2017)

\bibitem{bulatyang2018learn}
Bulat, A., Yang, J., Tzimiropoulos, G.: To learn image super-resolution, use a gan to learn how to do image degradation first.
\newblock European Conference on Computer Vision  (2018)

\bibitem{fan2023neural}
Fan, J., Liu, S., Ma, S., Zhou, H.M., Chen, Y.: Neural monge map estimation and its applications.
\newblock Transactions on Machine Learning Research  (2023).
\newblock \urlprefix\url{https://openreview.net/forum?id=2mZSlQscj3}.
\newblock Featured Certification

\bibitem{fritsche2019frequency}
Fritsche, M., Gu, S., Timofte, R.: Frequency separation for real-world super-resolution.
\newblock In: IEEE/CVF International Conference on Computer Vision (ICCV) Workshops (2019)

\bibitem{goodfellow2014generative}
Goodfellow, I., Pouget-Abadie, J., Mirza, M., Xu, B., Warde-Farley, D., Ozair, S., Courville, A., Bengio, Y.: Generative adversarial nets.
\newblock In: Advances in neural information processing systems, pp. 2672--2680 (2014)

\bibitem{gulrajani2017improved}
Gulrajani, I., Ahmed, F., Arjovsky, M., Dumoulin, V., Courville, A.C.: Improved training of {W}asserstein {GAN}s.
\newblock In: Advances in Neural Information Processing Systems, pp. 5767--5777 (2017)

\bibitem{he2016deep}
He, K., Zhang, X., Ren, S., Sun, J.: Deep residual learning for image recognition.
\newblock In: Proceedings of the IEEE conference on computer vision and pattern recognition, pp. 770--778 (2016)

\bibitem{heusel2017gans}
Heusel, M., Ramsauer, H., Unterthiner, T., Nessler, B., Hochreiter, S.: {GAN}s trained by a two time-scale update rule converge to a local nash equilibrium.
\newblock In: Advances in neural information processing systems, pp. 6626--6637 (2017)

\bibitem{huang2018multimodal}
Huang, X., Liu, M.Y., Belongie, S., Kautz, J.: Multimodal unsupervised image-to-image translation.
\newblock In: Proceedings of the European conference on computer vision (ECCV), pp. 172--189 (2018)

\bibitem{kantorovitch1958translocation}
Kantorovitch, L.: On the translocation of masses.
\newblock Management Science \textbf{5}(1), 1--4 (1958)

\bibitem{kim2020unsupervised}
Kim, G., Park, J., Lee, K., Lee, J., Min, J., Lee, B., Han, D.K., Ko, H.: Unsupervised real-world super resolution with cycle generative adversarial network and domain discriminator.
\newblock In: Proceedings of the IEEE/CVF Conference on Computer Vision and Pattern Recognition Workshops (CVPRW), pp. 1862--1871 (2020)

\bibitem{kingma2014adam}
Kingma, D.P., Ba, J.: Adam: A method for stochastic optimization.
\newblock arXiv preprint arXiv:1412.6980  (2014)

\bibitem{korotin2021neural}
Korotin, A., Li, L., Genevay, A., Solomon, J.M., Filippov, A., Burnaev, E.: Do neural optimal transport solvers work? a continuous wasserstein-2 benchmark.
\newblock Advances in Neural Information Processing Systems \textbf{34} (2021)

\bibitem{korotin2023kernel}
Korotin, A., Selikhanovych, D., Burnaev, E.: Kernel neural optimal transport.
\newblock In: The Eleventh International Conference on Learning Representations (2023)

\bibitem{korotin2023neural}
Korotin, A., Selikhanovych, D., Burnaev, E.: Neural optimal transport.
\newblock In: The Eleventh International Conference on Learning Representations (2023)

\bibitem{kotevski2009experimental}
Kotevski, Z., Mitrevski, P.: Experimental comparison of psnr and ssim metrics for video quality estimation.
\newblock In: International conference on ICT innovations, pp. 357--366. Springer (2009)

\bibitem{krizhevsky2012imagenet}
Krizhevsky, A., Sutskever, I., Hinton, G.E.: Imagenet classification with deep convolutional neural networks.
\newblock In: Proceedings of the 25th International Conference on Neural Information Processing Systems - Volume 1, NIPS'12, p. 1097–1105. Curran Associates Inc., Red Hook, NY, USA (2012)

\bibitem{lai2017deep}
Lai, W.S., Huang, J.B., Ahuja, N., Yang, M.H.: Deep laplacian pyramid networks for fast and accurate super-resolution.
\newblock In: Proceedings of the IEEE conference on computer vision and pattern recognition, pp. 624--632 (2017)

\bibitem{ledig2017photo}
Ledig, C., Theis, L., Husz{\'a}r, F., Caballero, J., Cunningham, A., Acosta, A., Aitken, A., Tejani, A., Totz, J., Wang, Z., et~al.: Photo-realistic single image super-resolution using a generative adversarial network.
\newblock In: Proceedings of the IEEE conference on computer vision and pattern recognition, pp. 4681--4690 (2017)

\bibitem{lim2017enhanced}
Lim, B., Son, S., Kim, H., Nah, S., Mu~Lee, K.: Enhanced deep residual networks for single image super-resolution.
\newblock In: Proceedings of the IEEE conference on computer vision and pattern recognition workshops, pp. 136--144 (2017)

\bibitem{pmlr-v202-liu23ai}
Liu, G.H., Vahdat, A., Huang, D.A., Theodorou, E., Nie, W., Anandkumar, A.: {I}$^2${SB}: Image-to-image schrödinger bridge.
\newblock In: A.~Krause, E.~Brunskill, K.~Cho, B.~Engelhardt, S.~Sabato, J.~Scarlett (eds.) Proceedings of the 40th International Conference on Machine Learning, \emph{Proceedings of Machine Learning Research}, vol. 202, pp. 22,042--22,062. PMLR (2023).
\newblock \urlprefix\url{https://proceedings.mlr.press/v202/liu23ai.html}

\bibitem{liu2019wasserstein}
Liu, H., Gu, X., Samaras, D.: Wasserstein {GAN} with quadratic transport cost.
\newblock In: Proceedings of the IEEE International Conference on Computer Vision, pp. 4832--4841 (2019)

\bibitem{liu2023unpaired}
Liu, H., Shao, M., Qiao, Y., Wan, Y., Meng, D.: Unpaired image super-resolution using a lightweight invertible neural network.
\newblock Pattern Recognition \textbf{144}, 109,822 (2023)

\bibitem{lu2020large}
Lu, G., Zhou, Z., Shen, J., Chen, C., Zhang, W., Yu, Y.: Large-scale optimal transport via adversarial training with cycle-consistency.
\newblock arXiv preprint arXiv:2003.06635  (2020)

\bibitem{lugmayr2019unsupervised}
Lugmayr, A., Danelljan, M., Timofte, R.: Unsupervised learning for real-world super-resolution.
\newblock 2019 IEEE/CVF International Conference on Computer Vision Workshop (ICCVW) pp. 3408--3416 (2019)

\bibitem{lugmayr2019aim}
Lugmayr, A., Danelljan, M., Timofte, R., Fritsche, M., Gu, S., Purohit, K., Kandula, P., Suin, M., Rajagoapalan, A., Joon, N.H., et~al.: Aim 2019 challenge on real-world image super-resolution: Methods and results.
\newblock In: 2019 IEEE/CVF International Conference on Computer Vision Workshop (ICCVW), pp. 3575--3583. IEEE (2019)

\bibitem{maeda2020unpaired}
Maeda, S.: Unpaired image super-resolution using pseudo-supervision.
\newblock In: Proceedings of the IEEE/CVF Conference on Computer Vision and Pattern Recognition, pp. 291--300 (2020)

\bibitem{mallasto2019q}
Mallasto, A., Frellsen, J., Boomsma, W., Feragen, A.: (q, p)-{W}asserstein {GAN}s: Comparing ground metrics for {W}asserstein {GAN}s.
\newblock arXiv preprint arXiv:1902.03642  (2019)

\bibitem{mroueh2017sobolev}
Mroueh, Y., Li, C.L., Sercu, T., Raj, A., Cheng, Y.: Sobolev gan.
\newblock arXiv preprint arXiv:1711.04894  (2017)

\bibitem{nilsson2020understanding}
Nilsson, J., Akenine-M{\"o}ller, T.: Understanding ssim.
\newblock arXiv preprint arXiv:2006.13846  (2020)

\bibitem{nowozin2016f}
Nowozin, S., Cseke, B., Tomioka, R.: f-{GAN}: Training generative neural samplers using variational divergence minimization.
\newblock In: Advances in neural information processing systems, pp. 271--279 (2016)

\bibitem{reibman2006quality}
Reibman, A.R., Bell, R.M., Gray, S.: Quality assessment for super-resolution image enhancement.
\newblock In: 2006 International conference on image processing, pp. 2017--2020. IEEE (2006)

\bibitem{ronneberger2015u}
Ronneberger, O., Fischer, P., Brox, T.: U-net: Convolutional networks for biomedical image segmentation.
\newblock In: International Conference on Medical image computing and computer-assisted intervention, pp. 234--241. Springer (2015)

\bibitem{rout2022generative}
Rout, L., Korotin, A., Burnaev, E.: Generative modeling with optimal transport maps.
\newblock In: International Conference on Learning Representations (2022).
\newblock \urlprefix\url{https://openreview.net/forum?id=5JdLZg346Lw}

\bibitem{saharia2022image}
Saharia, C., Ho, J., Chan, W., Salimans, T., Fleet, D.J., Norouzi, M.: Image super-resolution via iterative refinement.
\newblock IEEE transactions on pattern analysis and machine intelligence \textbf{45}(4), 4713--4726 (2022)

\bibitem{santambrogio2015optimal}
Santambrogio, F.: Optimal transport for applied mathematicians.
\newblock Birk{\"a}user, NY \textbf{55}(58-63), 94 (2015)

\bibitem{shrivastava2017learning}
Shrivastava, A., Pfister, T., Tuzel, O., Susskind, J., Wang, W., Webb, R.: Learning from simulated and unsupervised images through adversarial training.
\newblock In: Proceedings of the IEEE conference on computer vision and pattern recognition, pp. 2107--2116 (2017)

\bibitem{simonyan2014very}
Simonyan, K., Zisserman, A.: Very deep convolutional networks for large-scale image recognition.
\newblock arXiv preprint arXiv:1409.1556  (2014)

\bibitem{szegedy2016rethinking}
Szegedy, C., Vanhoucke, V., Ioffe, S., Shlens, J., Wojna, Z.: Rethinking the inception architecture for computer vision.
\newblock In: Proceedings of the IEEE conference on computer vision and pattern recognition, pp. 2818--2826 (2016)

\bibitem{taigman2016unsupervised}
Taigman, Y., Polyak, A., Wolf, L.: Unsupervised cross-domain image generation.
\newblock arXiv preprint arXiv:1611.02200  (2016)

\bibitem{taigman2017unsupervised}
Taigman, Y., Polyak, A., Wolf, L.: Unsupervised cross-domain image generation.
\newblock In: International Conference on Learning Representations (2017).
\newblock \urlprefix\url{https://openreview.net/forum?id=Sk2Im59ex}

\bibitem{villani2003topics}
Villani, C.: Topics in optimal transportation.
\newblock 58. American Mathematical Soc. (2003)

\bibitem{villani2008optimal}
Villani, C.: Optimal transport: old and new, vol. 338.
\newblock Springer Science \& Business Media (2008)

\bibitem{wang2021unsupervised}
Wang, W., Zhang, H., Yuan, Z., Wang, C.: Unsupervised real-world super-resolution: A domain adaptation perspective.
\newblock In: Proceedings of the IEEE/CVF International Conference on Computer Vision, pp. 4318--4327 (2021)

\bibitem{wang2009mean}
Wang, Z., Bovik, A.C.: Mean squared error: Love it or leave it? a new look at signal fidelity measures.
\newblock IEEE signal processing magazine \textbf{26}(1), 98--117 (2009)

\bibitem{Wei_2021_CVPR}
Wei, Y., Gu, S., Li, Y., Timofte, R., Jin, L., Song, H.: Unsupervised real-world image super resolution via domain-distance aware training.
\newblock In: Proceedings of the IEEE/CVF Conference on Computer Vision and Pattern Recognition (CVPR), pp. 13,385--13,394 (2021)

\bibitem{xie2019scalable}
Xie, Y., Chen, M., Jiang, H., Zhao, T., Zha, H.: On scalable and efficient computation of large scale optimal transport.
\newblock pp. 6882--6892. PMLR, Long Beach, California, USA (2019).
\newblock \urlprefix\url{http://proceedings.mlr.press/v97/xie19a.html}

\bibitem{yuan2018unsupervised}
Yuan, Y., Liu, S., Zhang, J., bing Zhang, Y., Dong, C., Lin, L.: Unsupervised image super-resolution using cycle-in-cycle generative adversarial networks.
\newblock 2018 IEEE/CVF Conference on Computer Vision and Pattern Recognition Workshops (CVPRW) pp. 814--81,409 (2018)

\bibitem{yue2023resshift}
Yue, Z., Wang, J., Loy, C.C.: Resshift: Efficient diffusion model for image super-resolution by residual shifting.
\newblock Advances in Neural Information Processing Systems \textbf{36}, 13,294--13,307 (2023)

\bibitem{zhang2018perceptual}
Zhang, R., Isola, P., Efros, A.A., Shechtman, E., Wang, O.: The unreasonable effectiveness of deep features as a perceptual metric.
\newblock In: CVPR (2018)

\bibitem{zhang2018unreasonable}
Zhang, R., Isola, P., Efros, A.A., Shechtman, E., Wang, O.: The unreasonable effectiveness of deep features as a perceptual metric.
\newblock In: Proceedings of the IEEE conference on computer vision and pattern recognition, pp. 586--595 (2018)

\bibitem{zhang2018rcan}
Zhang, Y., Li, K., Li, K., Wang, L., Zhong, B., Fu, Y.: Image super-resolution using very deep residual channel attention networks.
\newblock In: ECCV (2018)

\bibitem{zhou2020guided}
Zhou, Y., Deng, W., Tong, T., Gao, Q.: Guided frequency separation network for real-world super-resolution.
\newblock In: Proceedings of the IEEE/CVF Conference on Computer Vision and Pattern Recognition Workshops, pp. 428--429 (2020)

\bibitem{zhu2017unpaired}
Zhu, J.Y., Park, T., Isola, P., Efros, A.A.: Unpaired image-to-image translation using cycle-consistent adversarial networks.
\newblock In: Proceedings of the IEEE international conference on computer vision, pp. 2223--2232 (2017)

\end{thebibliography}

%%%%%%%%%%%%%%%%%%%%%%%%%%%%%%%%%%%%%%%%%%%%%%%%%%%%%%%%%%%%

%%%%%%%%%%%%%%%%%%%%%%%%%%%%%%%%%%%%%%%%%%%%%%%%%%%%%%%%%%%%

\newpage
\appendix

\section{Proofs}
\begin{proof}[Proof of Lemma \ref{lemma-optimal}] 
Assume that $T^{\lambda}$ is not an optimal map between $\mathbb{P}$ and $T^{\lambda}_{\#}\mathbb{P}$. Then there exists a more optimal $T^{\dagger}$ satisfying ${T^{\dagger}_{\#}\mathbb{P}=T_{\#}^{\lambda}\mathbb{P}}$ and $\mathcal{R}_{c}(T^{\dagger})<\mathcal{R}_{c}(T^{\lambda})$.
We substitute this $T^{\dagger}$ to \eqref{base-gan-c} and derive
\vspace{-1.5mm}
\begin{eqnarray}
\mathcal{D}(T^{\dagger}_{\#}\mathbb{P},\mathbb{Q})+\lambda \mathcal{R}_{c}(T^{\dagger})=\mathcal{D}(T^{\lambda}_{\#}\mathbb{P},\mathbb{Q})+
\lambda \mathcal{R}_{c}(T^{\dagger})<
\mathcal{D}(T^{\lambda}_{\#}\mathbb{P},\mathbb{Q})+\lambda \mathcal{R}_{c}(T^{\lambda}),
\nonumber
\end{eqnarray}
which is a contradiction, since $T^{\lambda}$ is a~minimizer of \eqref{base-gan-c}, but $T^{\dagger}$ provides the smaller value.
\end{proof}

\begin{proof}[Proof of Lemma \ref{corollary-equivalence}] 
We derive 
\vspace{-2.5mm}
\begin{eqnarray}
    \inf_{T:\mathcal{X}\mapsto\mathcal{Y}} \big[\mathcal{D}(T_{\#}\mathbb{P},\mathbb{Q})+\lambda \mathcal{R}_{c}(T) \big]
    = \inf_{T:\mathcal{X}\mapsto\mathcal{Y}} \big[\mathcal{D}(T_{\#}\mathbb{P},\mathbb{Q})+\lambda \int_{\mathcal{X}}c\big(x,T(x)\big)d\mathbb{P}(x)\big]=  \label{gan-problems-relation-line1}\\
    \inf_{T:\mathcal{X}\mapsto\mathcal{Y}} \big[\mathcal{D}(T_{\#}\mathbb{P},\mathbb{Q})+\lambda\cdot  \text{Cost} (\mathbb{P}, T_{\#}\mathbb{P})\big]
    = \inf_{\mathbb{Q}'\in\mathcal{P}(\mathcal{Y})}\big[\mathcal{D}(\mathbb{Q}', \mathbb{Q})+\lambda\cdot \text{Cost}(\mathbb{P}, \mathbb{Q}')\big].
    \label{gan-problems-relation-line2}
\label{gan-problems-relation}
\end{eqnarray}

In transition from \eqref{gan-problems-relation-line1} to \eqref{gan-problems-relation-line2}, we use the definition of OT cost \eqref{ot-primal-form-monge} and our Lemma \ref{lemma-optimal}, which states that the minimizer $T^{\lambda}$ of \eqref{base-gan-c} is an OT map, i.e., $\int_{\mathcal{X}}c\big(x,T^{\lambda}(x)\big)d\mathbb{P}(x)=\text{Cost}(\mathbb{P},T^{\lambda}_{\#}\mathbb{P})$. The equality in \eqref{gan-problems-relation-line2} follows from the fact that $\mathbb{P}$ is abs. cont. and $c(x,y)=\|x-y\|^{p}$: for all $\mathbb{Q}'\in\mathcal{P}(\mathcal{Y})$ there exists a (unique) solution $T$ to the Monge OT problem \eqref{ot-primal-form-monge} for $\mathbb{P}, \mathbb{Q}'$ \cite[\text{Thm. 1.17}]{santambrogio2015optimal}.
\end{proof}

\begin{proof}[Proof of Theorem \ref{lemma-biased}]
Let $\Delta \mathbb{Q}=\mathbb{P}-\mathbb{Q}$ denote the difference measure of $\mathbb{P}$ and $\mathbb{Q}$. It has zero total mass and $\forall\epsilon\in[0,1]$ it holds that $\mathbb{Q}+\epsilon\Delta\mathbb{Q}=\epsilon\mathbb{P}+(1-\epsilon)\mathbb{Q}$ is a mixture distribution of probability distributions $\mathbb{P}$ and $\mathbb{Q}$. As a result, for all $\epsilon\in[0,1]$, we have 
\vspace{-1.5mm}
\begin{eqnarray}
    \mathcal{F}(\mathbb{Q}+\epsilon\Delta \mathbb{Q}) = \mathcal{D}(\mathbb{Q}+\epsilon\Delta \mathbb{Q}, \mathbb{Q}) +  \lambda\cdot\text{Cost}(\mathbb{P}, \mathbb{Q}+\epsilon\Delta \mathbb{Q})=
    \nonumber
    \\
    \mathcal{D}(\mathbb{Q}, \mathbb{Q}) + o(\epsilon) +
    \lambda\cdot\text{Cost}(\mathbb{P}, \epsilon\mathbb{P}+(1-\epsilon)\mathbb{Q})
    \label{lemma2-before-villani-theorem} \leq
    \\
    o(\epsilon) + \lambda\cdot\epsilon \cdot\text{Cost}(\mathbb{P}, \mathbb{P}) + \lambda\cdot(1-\epsilon)\cdot \text{Cost}(\mathbb{P}, \mathbb{Q}) =
   o(\epsilon) + \lambda\cdot(1-\epsilon)\cdot \text{Cost}(\mathbb{P}, \mathbb{Q}) =
    \label{lemma2-before2-villani-theorem}
    \\
   \underbrace{\lambda\cdot \text{Cost}(\mathbb{P}, \mathbb{Q})}_{=\mathcal{F}(\mathbb{Q})} - \lambda \cdot\epsilon\cdot \underbrace{\text{Cost}(\mathbb{P}, \mathbb{Q})}_{>0} + o(\epsilon),
    \nonumber
\end{eqnarray}
where in transition from \eqref{lemma2-before-villani-theorem} to \eqref{lemma2-before2-villani-theorem}, we use $\mathcal{D}(\mathbb{Q},\mathbb{Q})=0$ and exploit the convexity of the OT cost \cite[Theorem 4.8]{villani2003topics}. In  \eqref{lemma2-before2-villani-theorem}, we use $\text{Cost}(\mathbb{P},\mathbb{P})=0$.
We see that $\mathcal{F}(\mathbb{Q}\!+\!\epsilon\Delta \mathbb{Q})$ is smaller then $\mathcal{F}(\mathbb{Q})$ for sufficiently small $\epsilon>0$, i.e.,
$\mathbb{Q}'\!=\!\mathbb{Q}$ does not minimize $\mathcal{F}$.
\end{proof}

\begin{proof}[Proof of Example \ref{example-toy-w2}]
Let $T(0)=t_{0}$ and $T(2)=t_{2}$. Then $T_{\#}\mathbb{P}=\frac{1}{2}\delta_{t_{0}}+\frac{1}{2}\delta_{t_{2}}$, and now  \eqref{base-gan-c} becomes
\vspace{-1.5mm}
$$\min_{t_{0},t_{2}}\bigg[\min\big\{\frac{1}{2}(t_{0}-1)^{2}+\frac{1}{2}(t_{2}-3)^{2}; \frac{1}{2}(t_{0}-3)^{2}+\frac{1}{2}(t_{2}-1)^{2}\big\}+\lambda \big\{\frac{1}{2}|0-t_0|+\frac{1}{2}|2-t_{2}|\big\}\bigg],$$
\vspace{-5mm}

where the second term is $\mathcal{R}_{c}(T)$ and the first term is the OT cost $\mathcal{D}(T_{\#}\mathbb{P}, \mathbb{Q})$ expressed as the minimum over the transport costs of two possible transport maps $t_0\mapsto 1; t_{2}\mapsto 3$ and $t_0\mapsto 3; t_{2}\mapsto 1$. The minimizer can be derived analytically and equals $t_{0}=1-\frac{\lambda}{2},t_{2}=3-\frac{\lambda}{2}$.
\end{proof}

\section{First Variations of GAN Discrepancies Vanish at the Optimum}
\label{sec-first-variation}
We demonstrate that the first variation of $\mathbb{Q}'\mapsto \mathcal{D}(\mathbb{Q}',\mathbb{Q})$ is equal to zero at $\mathbb{Q}'=\mathbb{Q}$ for common GAN discrepancies $\mathcal{D}$. This suggests that the corresponding assumption of our Theorem \ref{lemma-biased} is relevant.

To begin with, for a functional $\mathcal{G}:\mathcal{P}(\mathcal{Y})\rightarrow \mathbb{R}\cup \{\infty\}$, we recall the definition of its \textbf{first variation}. A measurable function $\delta\mathcal{G}[\mathbb{Q}]:\mathcal{Y}\rightarrow\mathbb{R}\cup \{\infty\}$ is called \textbf{the first variation} of $\mathcal{G}$ at a point $\mathbb{Q}\in\mathcal{P}(\mathcal{Y})$, if, for every measure $\Delta\mathbb{Q}$ on $\mathcal{Y}$ with zero total mass ($\int_{\mathcal{Y}}1\hspace{0.5mm} d\Delta\mathbb{Q}(y)=0$),  
\begin{equation}
\mathcal{G}(\mathbb{Q}+\epsilon \Delta\mathbb{Q})=\mathcal{G}(\mathbb{Q})+\epsilon\int_{\mathcal{Y}}\delta\mathcal{G}[\mathbb{Q}](y)\hspace{0.5mm}d\Delta\mathbb{Q}(y)+ o(\epsilon)
\label{first-variation-def}
\end{equation}
for all $\epsilon\geq 0$ such that $\mathbb{Q}+\epsilon\Delta\mathbb{Q}$ is a probability distribution. Here for the sake of simplicity we suppressed several minor technical aspects, see \cite[Definition 7.12]{santambrogio2015optimal} for details. Note that the first variation is defined \textbf{up to an additive constant}.

Now we recall the definitions of three most popular GAN discrepancies and demonstrate that their first variation is zero at an optimal point. We consider $f$-divergences \cite{nowozin2016f}, Wasserstein distances \cite{arjovsky2017wasserstein}.

\underline{\textbf{Case 1}} ($f$-divergence). Let $f:\mathbb{R}_{+}\rightarrow \mathbb{R}$ be a convex and differentiable function satisfying $f(1)=0$. The $f$-divergence between $\mathbb{Q}',\mathbb{Q}\in\mathcal{P}(\mathcal{Y})$ is defined by
\begin{equation}
    \mathcal{D}_{f}(\mathbb{Q}',\mathbb{Q})\stackrel{\normalfont{def}}{=}\int_{\mathcal{Y}}f\bigg(\frac{d\mathbb{Q}'(y)}{d\mathbb{Q}(y)}\bigg)d\mathbb{Q}(y).
    \label{f-divergence}
\end{equation}
The divergence takes finite value only if $\mathbb{Q'}\ll \mathbb{Q}$, i.e., $\mathbb{Q'}$ is absolutely continuous w.r.t.~$\mathbb{Q}$. Vanilla GAN loss \cite{goodfellow2014generative} is a case of $f$-divergence \cite[Table 1]{nowozin2016f}.

We define $\mathcal{G}(\mathbb{Q}')\stackrel{\normalfont def}{=}\mathcal{D}_{f}(\mathbb{Q}',\mathbb{Q})$. For $\mathbb{Q}'=\mathbb{Q}$ and some $\Delta\mathbb{Q}$ such that $\mathbb{Q}+\epsilon\Delta\mathbb{Q}\in\mathcal{P}(\mathcal{Y})$ we derive
\begin{eqnarray}
    \mathcal{G}(\mathbb{Q}+\epsilon\Delta\mathbb{Q})=\int_{\mathcal{Y}}f\bigg(\frac{d\mathbb{Q}(y)}{d\mathbb{Q}(y)}+\epsilon\frac{d\Delta\mathbb{Q}(y)}{d\mathbb{Q}(y)}\bigg)d\mathbb{Q}(y)=
    \int_{\mathcal{Y}}f\bigg(1+\epsilon\frac{d\Delta\mathbb{Q}(y)}{d\mathbb{Q}(y)}\bigg)d\mathbb{Q}(y)
    \label{before-taylor}
    \\
    =\int_{\mathcal{Y}}f(1)d\mathbb{Q}(y)+\int_{\mathcal{Y}}f'(1)\frac{d\Delta\mathbb{Q}(y)}{d\mathbb{Q}(y)}d\mathbb{Q}(y)+o(\epsilon)=
    \mathcal{G}(\mathbb{Q})+\int_{\mathcal{Y}}f'(1)d\Delta\mathbb{Q}(y)+o(\epsilon),
    \label{after-taylor}
\end{eqnarray}
where in transition from \eqref{before-taylor} to \eqref{after-taylor}, we consider the Taylor series w.r.t.~$\epsilon$ at $\epsilon=0$. We see that $\delta\mathcal{G}[\mathbb{Q}](y)\equiv f'(1)$ is constant, i.e., the first variation of $\mathbb{Q}'\mapsto\mathcal{D}_{f}(\mathbb{Q}',\mathbb{Q})$ vanishes at $\mathbb{Q}'=\mathbb{Q}$.

\underline{\textbf{Case 2}} (Wasserstein distance).
If in OT formulation \eqref{ot-primal-form} the cost function $c(x,y)$ equals $\|x-y\|^p$ with $p\geq1$, then $\big[\text{Cost}(\mathbb{P}, \mathbb{Q})\big]^{1/p}$ is called the \textit{Wasserstein distance} ($\mathbb{W}_p$). Generative models which use $\mathbb{W}_{p}^{p}$ as the discrepancy are typically called the Wasserstein GANs (WGANs). The most popular case is $p=1$ \cite{arjovsky2017wasserstein,gulrajani2017improved}, but more general cases appear in related work as well, see \cite{liu2019wasserstein,mallasto2019q}.

The first variation of $\mathcal{G}(\mathbb{Q}')\stackrel{\normalfont \textrm{def}}{=}\mathbb{W}_{p}^{p}(\mathbb{Q}',\mathbb{Q})$ at a point $\mathbb{Q}'$ is given by $\mathcal{G}[\mathbb{Q}'](y)=(f^{*})^{c}(y)$, where $f^{*}$ is the optimal dual potential (provided it is unique up to a constant) in \eqref{ot-dual-form-c} for a pair $(\mathbb{Q}',\mathbb{Q})$, see \cite[\wasyparagraph 7.2]{santambrogio2015optimal}. Our particular interest is to compute the optimal potential $(f^{*})^{c}$ at $\mathbb{Q}'=\mathbb{Q}$. We recall \eqref{ot-dual-form-c} and use $\mathbb{W}_{p}^{p}(\mathbb{Q},\mathbb{Q})=0$ to derive
$$\mathbb{W}_{p}^{p}(\mathbb{Q},\mathbb{Q})=0=
\sup_{f}\bigg[\int_{\mathcal{X}} f^{c}(y')d\mathbb{Q}'(y')+\int_{\mathcal{Y}} f(y)d\mathbb{Q}(y)\bigg].$$
One may see that $f^{*}\equiv 0$ attains the supremum (its $c$-transform $(f^{*})^{c}$ is also zero). Thus, \textbf{if $(f^{*})^{c}\equiv 0$ is a~unique potential} (up to a constant), the first variation of $\mathbb{Q}'\mapsto \mathbb{W}_{p}^{p}(\mathbb{Q}',\mathbb{Q})$ at $\mathbb{Q}'\!=\!\mathbb{Q}$ vanishes.

\vspace{-2mm}
\section{Experimental Details}
\subsection{Training Details}
\label{sec-details}
\vspace{-2mm}

 The practical optimization procedure of \textbf{Optimal Transport Solver} (OTS) is detailed in Algorithm \ref{algorithm-ot}.

 \begin{algorithm}[t!]
 \SetInd{0.5em}{0.3em}
     {
         \SetAlgorithmName{Algorithm}{empty}{Empty}
         \SetKwInOut{Input}{Input}
         \SetKwInOut{Output}{Output}
         \Input{distributions
         $\mathbb{P},\mathbb{Q}$ accessible by samples; mapping network $T_{\theta}:\mathcal{X}\rightarrow\mathcal{Y}$;\\ potential $f_{\omega}:\mathcal{X}\rightarrow\mathbb{R}$;
         transport cost $c:\mathcal{X}\times\mathcal{Y}\rightarrow\mathbb{R}$; number $K_{T}$ of inner iters;\\
         }
         \Output{approximate OT map  $(T_{\theta})_{\#}\mathbb{P}\approx \mathbb{Q}$\;}
         }
        
         \Repeat{not converged}{
             Sample batches $X\sim\mathbb{P}$, $Y\!\sim\! \mathbb{Q}$\;
             $\mathcal{L}_{f}\leftarrow  \frac{1}{|Y|}\sum\limits_{y\in Y}f_{\omega}(y) - \frac{1}{|X|}\sum\limits_{x\in X}f_{\omega}\big(T_{\theta}(x)\big)$\;
             Update $\omega$ by using $\frac{\partial \mathcal{L}_{f}}{\partial \omega}$ to maximize $\mathcal{L}_{f}$\;
            
             \For{$k_{T} = 1,2, \dots, K_{T}$}{
                 Sample batch $X\sim \mathbb{P}$\;
                 ${\mathcal{L}_{T}\leftarrow\frac{1}{|X|}\sum\limits_{x\in X}\big[c
                 \big(x, T_{\theta}(x)\big)-f_{\omega}\big(T_{\theta}(x)\big)\big]}$\;
             Update $\theta$ by using $\frac{\partial \mathcal{L}_{T}}{\partial \theta}$ to minimize $\mathcal{L}_{T}$\;
             }
         }
        
         \caption{ OT solver to compute the OT map between $\mathbb{P}$ and $\mathbb{Q}$ for transport cost $c(x,y)$.}
         \label{algorithm-ot}
 \end{algorithm}

\noindent\textbf{Perceptual cost.}
In \ref{sec-experiments-aim} we test following \textit{perceptual cost} as $b$:
\begin{eqnarray}b(x^{\text{up}}\hspace{-1mm},y)\hspace{-0.7mm}=\hspace{-0.7mm}\text{MSE}(x^{\text{up}}\hspace{-1mm},y)\hspace{-0.7mm}+\hspace{-0.7mm}\nicefrac{1}{3}\cdot\text{MAE}(x^{\text{up}}\hspace{-1mm},y)\hspace{-0.7mm}+\hspace{-0.7mm}
\nicefrac{1}{50} \hspace{-0.7mm}\cdot\hspace{-0.7mm} \hspace{-6mm}\sum_{k\in\{3,8,15,22\}}\hspace{-6mm}\text{MSE}\big(f_{k}(x^{\text{up}}),f_{k}(y)\big),
\nonumber
\label{perceptual-cost}
\end{eqnarray}

\vspace{-3.5mm}
\looseness=-1
where $f_{k}$ denotes the features of the $k$th layer of a pre-trained VGG-16 network \cite{simonyan2014very}, MAE is the mean absolute error $\text{MAE}(x,y)=\frac{\|x-y\|_{1}}{\dim (\mathcal{Y})}$.

\noindent\textbf{Dynamic transport cost}. In the preliminary experiments, we used bicubic upsampling as the ``$\text{Up}$" operation. Later, we found that the method works better if we gradually change the upsampling. We start from the bicubic upsampling. Every $k_{c}$ iterations of $f_{\omega}$ (see Table \ref{table-params}), we change the cost to $c(x,y)=b\big(T_{\theta}'(x), y\big)$, where $T_{\theta}'$ is a fixed frozen copy of the currently learned SR map $T_{\theta}$.

\noindent\textbf{Hyperparameters.} For EDSR, we set the number of residual blocks to 64, the number of features to 128, and the residual scaling to 1. For UNet, we set the base factor to 64. The training details are given in Table \ref{table-params}. We provide a comparison of the hyperparameters of FSSR, DASR and OTS in Table \ref{table-comparison-hyperparams}. In contrast to FSSR and DASR, \underline{OTS method does not contain a degradation part}. This helps to notably reduce the amount of tunable hyperparameters.

\noindent\textbf{Optimizer.} We employ Adam \cite{kingma2014adam}. 

\noindent\textbf{Computational complexity}. Training OTS with EDSR as the transport map and the perceptual transport cost on AIM 2019 dataset takes $\approx 4$ days on a single Tesla V100 GPU.

\begin{table*}[!h]
\centering
\scriptsize
\addtolength{\tabcolsep}{-1.3mm}
\hspace*{-6mm}\begin{tabular}{c|c|c|c|c|c|c|c|c|c|c|c}\hline
\textbf{Experiment} & $\text{dim}(\mathcal{X})$ & $\text{dim}(\mathcal{Y})$ & $f$ & $T$ & $k_{T}$ & $lr_{f}$ & $lr_{T}$ & \makecell{\textbf{Initial}\\\textbf{cost}} & \makecell{\textbf{Total }\\ \textbf{iters} ($f$)} & \makecell{\textbf{Cost}\\\textbf{update}\\\textbf{every}} & \makecell{\textbf{Batch}\\\textbf{size}} \\ \hline
\makecell{Benchmark\\ (\wasyparagraph\ref{sec-bias-experiments})} & $3\times 64\times 64$ & $3\times 64\times 64$ & \multirow{6}{*}{ResNet} & UNet & 10 & \multirow{4}{*}{$10^{-4}$} & \multirow{4}{*}{$10^{-4}$} & MSE & 10K & $-$ & 64 \\ 
 \cline{1-3}\cline{5-6}\cline{9-12}
\multirow{2}{*}{\makecell{AIM-19\\(\wasyparagraph\ref{sec-experiments-aim})}} & \multirow{2}{*}{\makecell{$3\times 32\times 32$\\(patches)}} & \multirow{2}{*}{\makecell{$3\times 128\times 128$\\(patches)}} &  & EDSR & 15 &   &   & \makecell{Bicubic +\\ MSE} & 50K & 25K &  8\\ \cline{5-6}\cline{9-12}
&  &  &  & EDSR & 10 &   &   & \makecell{Bicubic +\\ VGG} & 50K & 20K & 8 \\
\hline
\end{tabular}
\caption{\centering Hyperparameters that we use in the experiments with OTS Algorithm \ref{algorithm-ot}. }
\label{table-params}
\end{table*}
\vspace{-2mm}
\begin{table*}[!h]
\centering
\footnotesize\addtolength{\tabcolsep}{-1mm}
\hspace*{-6mm}\begin{tabular}{c|c|c|c}\toprule
 \textbf{Method} & \makecell{\textbf{Degradation part}} & \makecell{\textbf{Super-resolution part}} & \textbf{Total} \\ \midrule
\textbf{FSSR} & \makecell{2 neural networks; \\ 2 optimizers; \\ 2 schedulers; \\ 1 adversarial loss; \\ 1 content loss ($\ell_1$+perceptual)} & \makecell{2 neural networks; \\ 2 optimizers; \\ 2 schedulers; \\ 1 adversarial loss; \\ 1 content loss ($\ell_1$+perceptual)} & \makecell{4 neural networks; \\ 4 optimizers; \\ 4 schedulers; \\ 2 adversarial losses; \\ 2 content losses ($\ell_1$+perceptual)} \\ \midrule
\textbf{DASR} & \makecell{2 neural networks; \\ 2 optimizers; \\ 2 schedulers; \\ 1 adversarial loss; \\ 1 content loss ($\ell_1$+perceptual)} & \makecell{2 neural networks; \\ 2 optimizers; \\ 2 schedulers; \\ 1 adversarial loss; \\ 1 content loss ($\ell_1$+perceptual)} & \makecell{4 neural networks; \\ 4 optimizers; \\ 4 schedulers; \\ 2 adversarial losses; \\ 2 content losses ($\ell_1$+perceptual)} \\ \midrule
\makecell{\textbf{OTS}\\} & $-$ & \makecell{2 neural networks; \\ 2 optimizers; \\ 1 cost ($\ell_2$+$\ell_1$+perceptual)} & \makecell{2 neural networks; \\ 2 optimizers; \\ 1 cost ($\ell_2$+$\ell_1$+perceptual)} \\ \bottomrule
\end{tabular}
\caption{\centering Comparison of hyperparameters used in FSSR, DASR and OTS methods.}
\label{table-comparison-hyperparams}
\end{table*}

\vspace{-10mm}
\subsection{Evaluation metrics}
\label{sec-details-eval}
We use several evaluation metrics: Peak Signal-to-Noise Ratio (PSNR), Structural Similarity Index Measure (SSIM), Learned Perceptual Image Patch Similarity (LPIPS), and Fréchet Inception Distance (FID). These metrics assess different aspects of image quality.

\noindent \textbf{PSNR} measures pixel-wise similarity between the images and is given by the formula: $\text{PSNR}(x,y)=10\cdot \log_{10}\frac{m}{\text{MSE}(x,y)}$ where $x$ and $y$ correspond to the given images, $m$ - maximum possible pixel value of the image. Thus, it is applicable for assessing the average pixel-wise similarity of the images and, thus, usually favors blurry images \cite{wang2009mean}.

\noindent\textbf{SSIM} metric is designed to measure the structural differences between the images and is given by the formula    
\vspace{-4mm}$$
\text{SSIM}(x, y) = \frac{(2 m_x m_y + C_1)(2\sigma_{xy} + C_2)}{(m_x^2 + m_y^2 + C_1)(\sigma_x^2 + \sigma_y^2 + C_2)}.
$$\vspace{-3mm}

\noindent Here $x,y$ correspond to the given images, $m_x,m_y,\sigma^2_x,\sigma^2_y,\sigma_{xy}$ denote the means, variances and covariances of pixel values, respectively; constants $C_1,C_2$ are given by $C_1=(0.01\cdot L)^2,C_2=(0.03\cdot L)^2$ where $L$ denotes the range of pixel values, e.g., 255. 
While SSIM can capture structural changes of images, e.g., blurring, noise addition, it is still not sensitive to the other types of changes, e.g., in brightness or contrast \cite{kotevski2009experimental}. Overall, PSNR and SSIM are not well-aligned with visual quality of images, see \cite{reibman2006quality}.

The next two metrics are based on the usage of pre-trained neural networks and are known to better capture the perceptual quality of the images.

\noindent\textbf{LPIPS} \cite{zhang2018unreasonable} metric compares the images based on their feature embeddings instead of the pixel values. The embeddings are retrieved from the pre-trained neural networks, e.g., AlexNet \cite{krizhevsky2012imagenet}. Then the metric is calculated as a weighted sum of $\ell^2$-distances between the embeddings from all layers of the network. In contrast to PSNR and SSIM which deal with pixel-wise similarity of images, LPIPS captures high-level texture and content similarity and, thus, is better aligned with the human perception of image quality. However, all these metrics compare generated and true images which are assumed to be given in \textit{pairs}.

\noindent\textbf{FID} \cite{heusel2017gans} metric is calculated using the \textit{unpaired datasets} of generated and true images. It uses deep neural networks, i.e., Inception-v3 \cite{szegedy2016rethinking}, to retrieve the feature embeddings for the given datasets of true and generated images. Then the mean vectors $\mu_1,\mu_2$ and covariance matrices $\Sigma_1,\Sigma_2$ of these feature vectors are used to compute  the final score as
$$
\text{FID} = \left\| \mu_1 - \mu_2 \right\|^2 + \text{Tr}\left( \Sigma_1 + \Sigma_2 - 2\left(\Sigma_1 \Sigma_2\right)^{1/2} \right).
$$
This metric is sensitive to perceptual quality and realism of images \cite{borji2019pros,heusel2017gans}. Still, both LPIPS and FID depend on the underlying pre-trained neural networks and might be less informative for the images which significantly differ from the datasets used to train that networks.

In our paper, we report all of the stated metrics but mostly focus our attention on the perceptual metrics which better correlate with human perception of image quality. We calculate the metrics using \texttt{scikit-image} for SSIM and open source implementations for PSNR\footnote{\url{github.com/photosynthesis-team/piq}}, LPIPS\footnote{\url{github.com/richzhang/PerceptualSimilarity}} and FID\footnote{\url{github.com/mseitzer/pytorch-fid}}. 

\vspace{-2mm}
\section{Assessing the bias of methods on AIM19 dataset} 
\label{sec-aim-bias}
We additionally demonstrate the bias issue by comparing color palettes of HR images and super-resolution results of different methods, see Figure \ref{fig:aim-bias}. We construct palettes by choosing random image pixels  from dataset images and representing them as an RGB point cloud in $[0,1]^{3}\subset\mathbb{R}^{3}$. Figure~\ref{fig:aim-bias} shows that OTS {\color{LimeGreen}\textbf{(d)}} captures \textit{large contrast} of HR {\color{LimeGreen}\textbf{(a)}} images (variance of its palette), while FSSR {\color{red}\textbf{(e)}}, DASR {\color{red}\textbf{(f)}}, Bicubic Upscale {\color{red}\textbf{(c)}} palettes are \textit{less contrastive} and closer to LR {\color{orange}\textbf{(b)}}. We construct palettes 100 times to evaluate their average contrast (variance). The metric \textit{quantitatively} confirms that OTS method better captures the contrast of HR dataset, while GAN-based methods (FSSR and DASR) are notably {\textit{biased}} towards LR dataset statistics (low contrast).
\begin{figure}[h!]
{
\begin{center}
\includegraphics[width=0.9\linewidth]{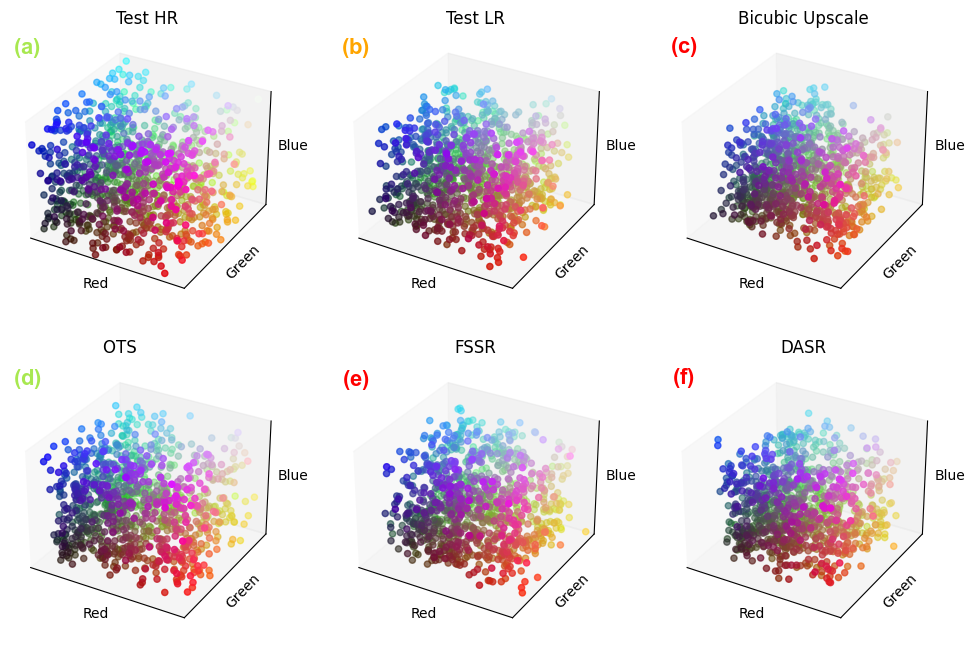}
\end{center}
}
\centering
\small
{
        \small
        \begin{tabular}{c|c|c|c|c|c|c}
        \toprule
        \textit{Dataset} &  \textit{Test HR} & \textit{Test LR} &  \textit{Bicubic} & \textit{OTS (VGG)} &  \textit{FSSR}    &  \textit{DASR} \\
        \midrule
        \textbf{Variance} & \makecell{{\color{LimeGreen}\textbf{0.24}} \\ $\pm0.01$}  & \makecell{{\color{orange}0.17} \\ $\pm0.01$} & \makecell{{\color{red}0.15} \\ $\pm0.02$} &  \makecell{{\color{LimeGreen}\textbf{0.20}} \\ $\pm0.03$} &  \makecell{{\color{red}0.17} \\ $\pm0.02$}  & \makecell{{\color{red}0.15} \\ $\pm0.02$} \\
        \bottomrule
        \end{tabular}
        }
\caption{\centering Color palettes and their variance for Test HR, LR datasets and solutions of  Bicubic Upscale, OTS, FSSR, DASR methods on AIM19.}
\label{fig:aim-bias}
\end{figure}
 \section{Connection between GAN objectives and Equation \eqref{base-gan-c}}
\label{gan-objectives-example}

Typical objectives of GAN-based approaches consist of multiple losses $-$ usually one adversarial and several content losses. To make the exposition simple, in our paper, we represented all the content losses as a single loss $c(\cdot, \cdot)$. Below we provide several examples showing how the objectives of popular GAN-based approaches to unpaired image SR could be viewed as \eqref{base-gan-c}. For all of these methods, our Lemma \ref{lemma-optimal} applies without any changes. We include in brackets the number of papers citations according to Google Scholar to show that chosen methods are widely used.

\noindent\textbf{FaceSR} (2018, 493 citations) The paper of \cite{bulatyang2018learn} presents one of the first GAN-based approaches to unpaired image SR problem. The method is composed of two steps. First, it learns a degradation between unpaired HR and LR images. Then it employs a second GAN to learn a supervised mapping between paired generated LR and corresponding HR images. The objective of the unpaired step (see their Eq. (1)) is as follows:
\begin{equation*}
    l = \underbrace{\alpha l_\text{pixel}}_\text{content loss} + \underbrace{\beta l_{\text{GAN}}.}_\text{adversarial loss}
\end{equation*}
Here $l_\text{pixel}$ is the MSE loss between the generated LR image and downsampled HR. Thus, the objective of this method exactly follows Equation \eqref{base-gan-c}.

\noindent\textbf{CinCGAN} (2018, 904 citations) The method of \cite{yuan2018unsupervised} is an other pioneering GAN-based approach to unpaired image SR problem, which establishes a different to FaceSR group of two-step methods. First, it uses one CycleGAN to learn a mapping between given noisy LR images and downsampled HR ("clean LR") images. Then, a second CycleGAN fine-tunes a mapping between real LR and HR images. The objective for the first GAN (see their Eq. (5)) is as follows:
\begin{eqnarray*}
    \mathcal{L}^{LR}_\text{total} = \underbrace{\mathcal{L}^\text{LR}_\text{GAN}}_\text{adversarial loss} +\underbrace{ w_1 \mathcal{L}^\text{LR}_\text{cyc} + w_2 \mathcal{L}^\text{LR}_\text{idt} + w_3 \mathcal{L}^\text{LR}_\text{TV}}_\text{content loss}.
\end{eqnarray*}
Here $\mathcal{L}^{LR}_\text{cyc}$ is the cycle-consistency loss\footnote{$\mathcal{L}^\text{LR}_\text{cyc}$ is defined as the MSE loss between given LR image $x$ and $G_2(G_1(x))$, where $G_1$ learns to map real LR images to "clean" ones and $G_2$ learns an opposite mapping. For a fixed $G_2$ this loss can be considered as a part of the content loss.}, $\mathcal{L}^\text{LR}_\text{idt}$ $-$ $l_1$ identity loss, $\mathcal{L}^\text{LR}_\text{TV}$ $-$ total variation loss.

\noindent\textbf{FSSR} (Winner of the AIM Challenge on Real-World SR \cite{lugmayr2019aim}, 2019, 260 citations)
FSSR \cite{fritsche2019frequency} method employs a similar to FaceSR strategy. It firstly learns a mapping between downsampled HR images and given unpaired LR images, and then uses the generated pairs to learn a supervised SR model. 
The objective of the unpaired step (see their Eq. (6)) is defined by:
\begin{equation*}
    \mathcal{L}_d = \underbrace{0.005 \mathcal{L}_{\text{tex, d}}}_\text{adversarial loss} + \underbrace{\mathcal{L}_{\text{col, d}} + 0.01 \mathcal{L}_{\text{per, d}}}_\text{content loss},
\end{equation*}
where the texture (adversarial) loss $\mathcal{L}_{\text{tex, d}}$ and the color ($l_1$ identity) loss $\mathcal{L}_{\text{col, d}}$ are applied to low frequencies of the images, while the perceptual loss $\mathcal{L}_{\text{per, d}}$ $-$ to the features of the full images.

\noindent\textbf{DASR} (2021, 165 citations) DASR \cite{Wei_2021_CVPR} structure is also based on the similar to FSSR principles and its two-step structure.
In contrast to FSSR, a SR network is trained in a partially supervised manner using not only generated, but also real LR images.
The objective of the fully unpaired degradation learning step (see their Eq. (4)) is as follows:
\begin{eqnarray*}
    \mathcal{L}_\text{DSN} = \underbrace{\alpha \mathcal{L}_\text{con} + \beta \mathcal{L}_\text{per}}_\text{content loss} + \underbrace{\gamma \mathcal{L}_\text{adv}^G.}_\text{adversarial loss}
\end{eqnarray*}
Here the adversarial loss $\mathcal{L}_\text{adv}^G$ is defined on high frequencies of the image, while the content $\mathcal{L}_\text{con}$ ($l_1$ identity) and the perceptual $\mathcal{L}_\text{per}$ losses are defined on full images and their features respectively.

\noindent\textbf{ESRGAN-FS} (2020, 56 citations) ESRGAN-FS is an other two-step approach based on the principle of learning the degradation, see \cite{zhou2020guided}. 
The objective of its unpaired degradation learning step (see their Eq. (4)) is as follows:
\begin{eqnarray*}
    \mathcal{L}_\text{total} = \underbrace{\lambda_{t1} \cdot \mathcal{L}_\text{low} + \lambda_{t2} \cdot \mathcal{L}_\text{per}}_\text{content loss} + \underbrace{\lambda_{t3} \cdot \mathcal{L}_\text{high}}_\text{adversarial loss}.
\end{eqnarray*}
Here $\mathcal{L}_\text{low}$ ($l_1$ identity) loss is applied to low frequencies of the images, the perceptual loss $\mathcal{L}_\text{per}$ $-$ to the features of the full images, while $\mathcal{L}_\text{high}$ (adversarial loss) $-$ high frequencies of the images.

\newpage
\section{Additional Qualitative Results on AIM19}
\label{sec-additional-results}
\begin{figure*}[!h]
\centering
\begin{subfigure}{\textwidth}
\centering\includegraphics[width=0.9\linewidth]{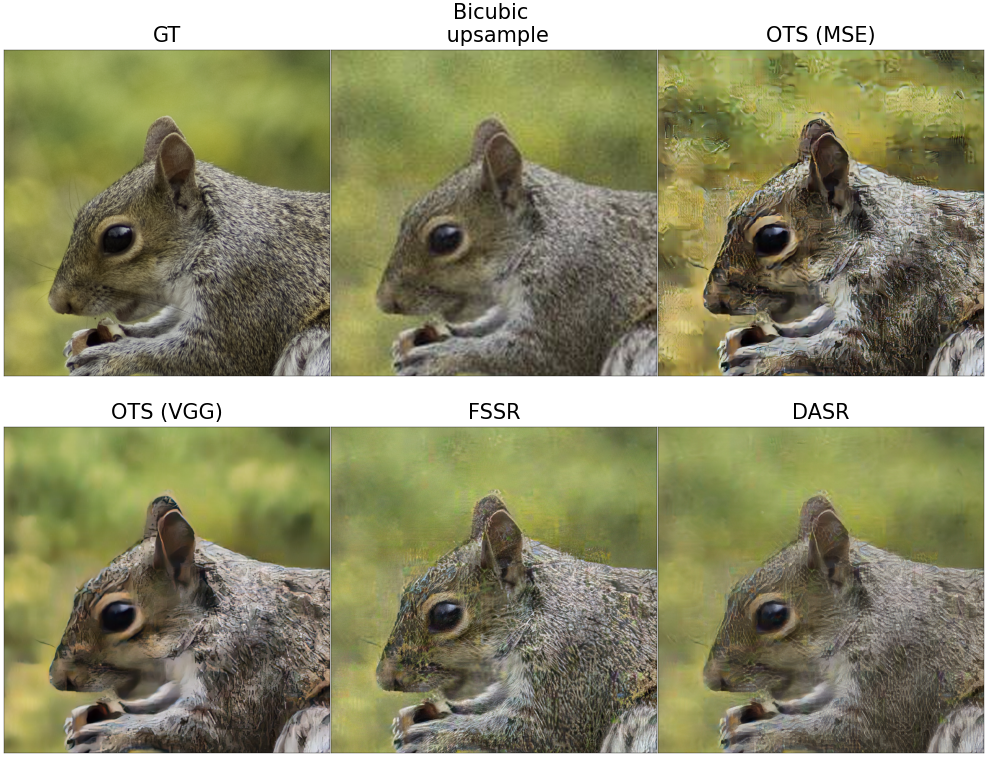}
\end{subfigure}
\begin{subfigure}{\linewidth}
\vspace{3mm}
\centering\includegraphics[width=0.9\textwidth]{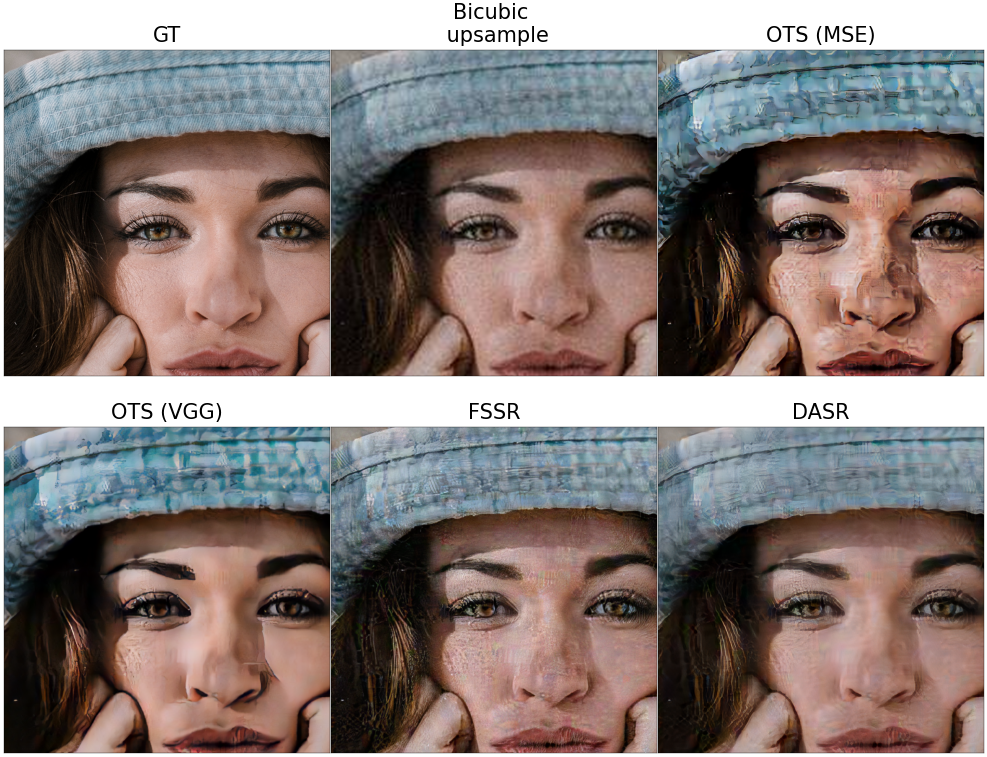}
\end{subfigure}
\caption*{\centering Figure 9: Additional qualitative results of OTS, bicubic upsample, FSSR and \protect{\linebreak} DASR on AIM 2019 (800$\times$800 crops).
}
\vspace{-8mm}
\label{fig:aim_ex2}
\end{figure*}

\begin{figure*}[!t]
\centering
\begin{subfigure}{\textwidth}
\includegraphics[width=0.9\linewidth]{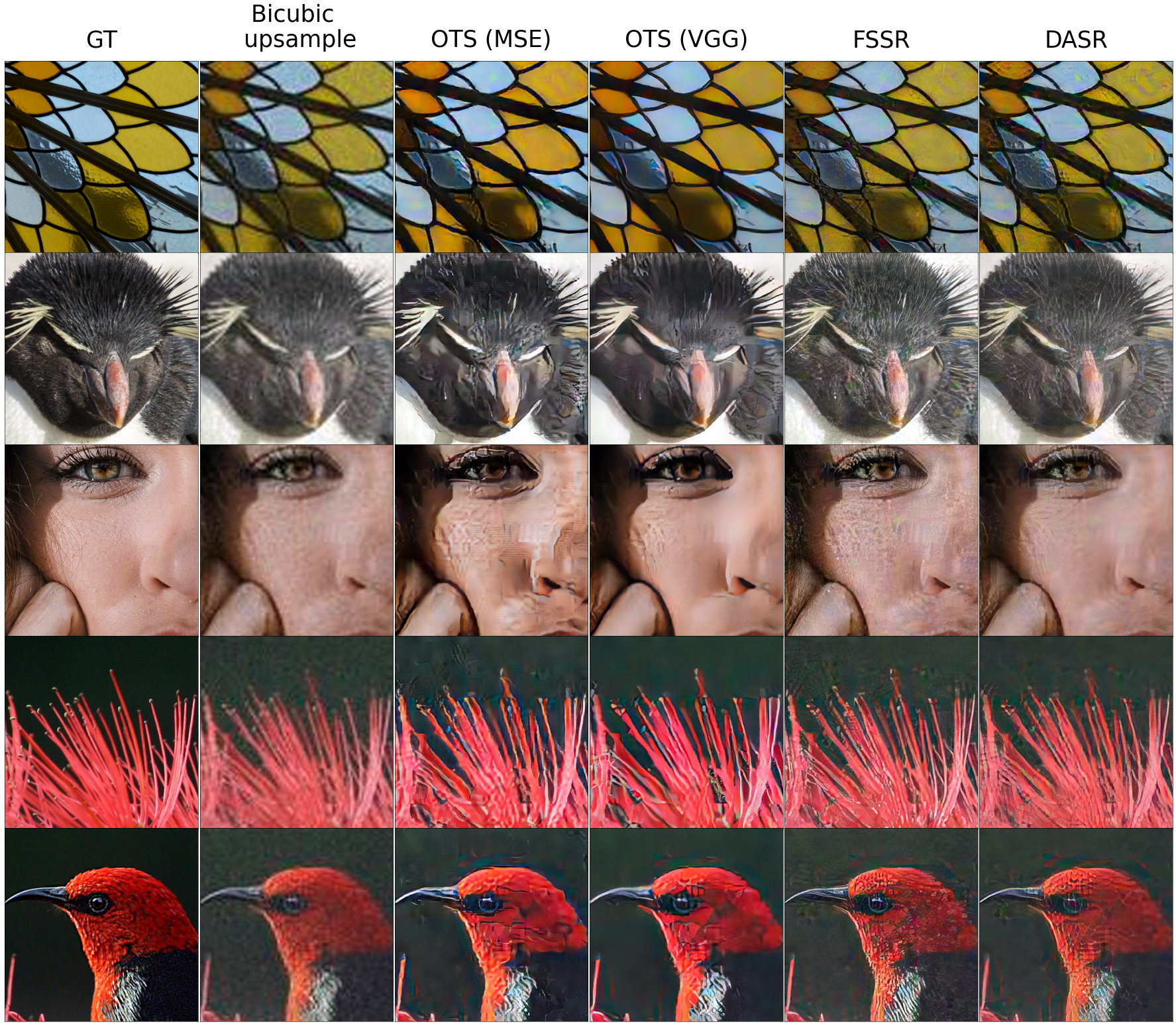}
\end{subfigure}
\begin{subfigure}{\textwidth}
\vspace{3mm}
\includegraphics[width=0.9\linewidth]{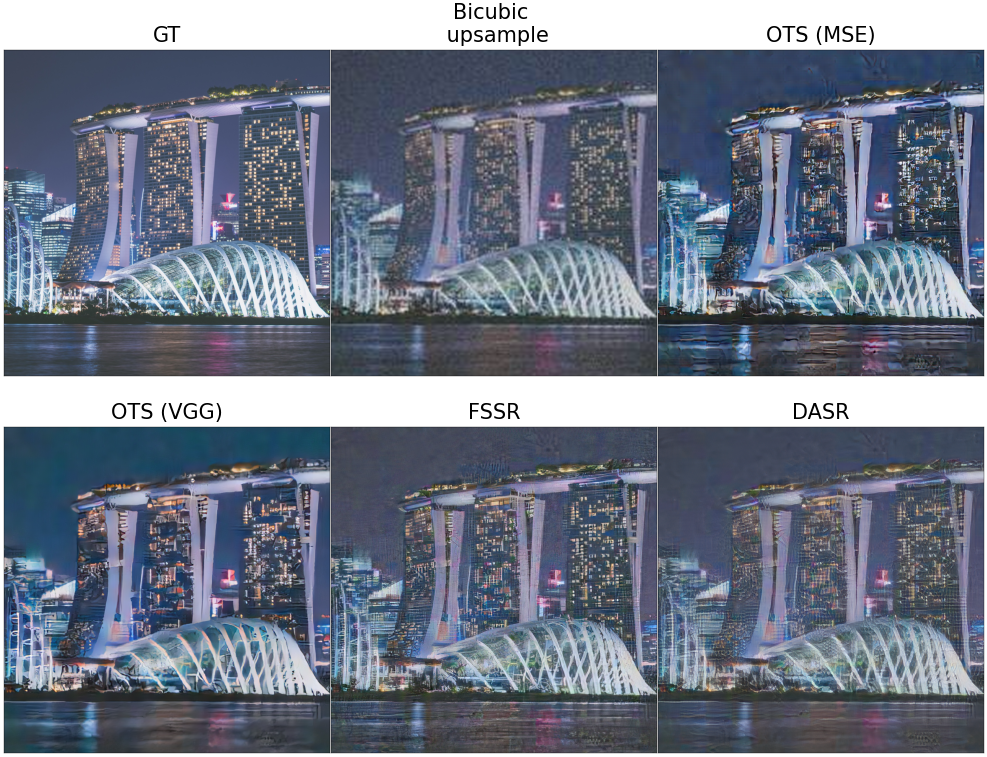}
\end{subfigure}
\caption*{\centering Figure 10: Additional qualitative results of OTS, bicubic upsample, FSSR and DASR on AIM 2019. The sizes of crops on the 1st and 2nd images are 350$\times$350 and 800$\times$800, respectively.
}
\vspace{-5mm}
\label{fig:aim_res2}
\end{figure*}

\end{document}